\documentclass[format=acmsmall, review=false]{acmart}
\usepackage{acm-ec-26}
\usepackage{placeins}
\usepackage{booktabs} 
\usepackage[ruled]{algorithm2e} 

\SetAlFnt{\small}
\SetAlCapFnt{\small}
\SetAlCapNameFnt{\small}
\SetAlCapHSkip{0pt}
\IncMargin{-\parindent}

\setcitestyle{authoryear}

\title{A Dynamic Equilibrium Model for Automated Market Makers}

\author{Chengqi Zang, Zhenghui Wang, Weitong Zhang}

\begin{abstract}
Automated Market Makers (AMMs) are a central component of decentralized exchanges, yet their equilibrium foundations and microeconomic mechanisms remain incompletely understood. This paper develops a dynamic equilibrium framework for Constant Function Market Makers (CFMMs) that formalizes the strategic interaction between arbitrageurs and liquidity providers (LPs) over time. 
We make three main contributions. First, we derive and empirically validate an intrinsic buy–sell asymmetry in CFMM price impact. Even in the absence of directional price movements, the geometric structure of constant-product AMMs implies systematically different execution costs for buying and selling, a prediction that we confirm using on-chain transaction data. Second, we characterize the optimization problems of arbitrageurs and LPs in closed form, incorporating slippage and fees. In a baseline environment with only informed arbitrageurs, we show that providing liquidity is strictly dominated for LPs: arbitrage-driven price corrections generate negative jump returns that cannot be offset by fees, yielding a degenerate equilibrium with minimal liquidity provision. Third, motivated by empirical evidence, we extend the model to include agent heterogeneity, endogenous gas fees, and time-varying volatility. 
In this extended environment, noise trading, arbitrage races, and execution costs jointly determine LP returns, giving rise to an interior equilibrium in which optimal liquidity provision is non-monotonic in volatility and exhibits a hump-shaped relationship.
Overall, this paper builds a dynamic equilibrium model calibrated on extensive data that characterize the complex interaction between informed arbitrageurs, noise traders, and liquidity providers.
\end{abstract}

\begin{document}


\begin{titlepage}

\maketitle


\end{titlepage}

\vspace{1cm}
\setcounter{tocdepth}{2} 

\section{Introduction}

Automated market makers (AMMs) have become core trading infrastructure in decentralized finance (DeFi). 
On AMM-based decentralized exchanges (DEXs), liquidity is supplied endogenously by liquidity providers (LPs) who deposit two assets into a smart contract and earn a fee from each swap. 
In 2023, Uniswap---the largest AMM DEX---processed average daily volume of \$1.3B and cumulative volume of nearly \$1.65T~\cite{adams2024dontletmevslip}, and by 2024 the AMM-dominated DEX sector reached roughly \$14.4B in total value locked~\cite{defillama2024tvl}. 
Despite this scale, the economic foundations of liquidity provision in AMMs remain poorly understood.

The central economic puzzle is straightforward to state.
In constant-product AMMs, LPs passively absorb all incoming order flow while prices are continuously pulled toward an external reference price—typically formed on centralized exchanges—by arbitrage.
This exposes LPs to adverse selection and impermanent loss  from informed trading.
Nevertheless, liquidity is provided at scale and persists across market regimes.
Understanding why LPs provide liquidity, and how this can be sustained in equilibrium, remains an open question.

A growing theoretical literature has made important progress toward answering this question.
Recent work studies AMM price dynamics, arbitrage behavior, and LP wealth evolution under various informational and institutional assumptions, providing tractable equilibrium characterizations and useful benchmarks.
For example, \cite{milionis2023automated}~characterize arbitrage profits and stationary price dispersion in the presence of LP fees; \cite{routledge2024automated}~analyze static equilibria with informed and uninformed traders; and \cite{aoyagi2021coexisting}~study liquidity provision in environments where AMMs coexist with limit order books.
These papers establish a rigorous starting point for equilibrium analysis of AMMs.

At the same time, most existing models rely on exogenous trading primitives—in particular, a stable flow of uninformed or “noise” trades that compensates LPs for adverse selection.
As a result, while these frameworks generate internally consistent equilibria, they leave open whether their key assumptions align with observed on-chain behavior.

We first develop a dynamic equilibrium model with a constant-product AMM, proportional LP fees, and only profit-maximizing informed arbitrageurs trading against an exogenous reference price~(\S4).
This baseline model delivers two implications.
First, the AMM pricing rule generates a structural asymmetry between buy-side and sell-side arbitrage.
Second, fee revenue alone is insufficient to offset adverse selection and price impact caused by informed arbitrageurs.

We then take these predictions to the data~(\S5).
Using transaction-level data from multiple Uniswap pools, we empirically verify the model’s predicted buy–sell asymmetry in arbitrage behavior.
By classifying transactions into profitable and unprofitable by comparing the executing price to an outside reference price, we found that across all pools we study, a majority of swaps are ex-post unprofitable relative to fee-adjusted centralized exchange prices.
In some pools, over 60\% of transactions fall into this category, and most pools exhibit rates near 80\%.
While many of these unprofitable trades are small in size, a nontrivial fraction are large, making it implausible to interpret them as a stable background flow noise trading.
Moreover, these ex-post unprofitable trades are not randomly distributed.
We show that centralized exchange price volatility predicts the volume of ex-post unprofitable AMM trades and gas fees.
Together, these patterns point to a volatility-driven feedback mechanism linking (failed)arbitrage activity, congestion, and transaction costs.

Motivated by these findings, we extend the baseline model along three dimensions~(\S6).
(i) variable CEX price volatility, (ii) gas fee driven by the CEX price volatility, and (iii) a micro-founded arbitrage race to explain the unevenly distributed unprofitable transactions: when arbitrage opportunity arrives, multiple arbitrageurs attempt to exploit the same price discrepancies, but only the earliest transactions capture the spread and lost arbitrager trade against an already-adjusted pool thus resulted in a expost unprofitable transaction.
One important result from the extended model is that it restores an interior solution for liquidity provision: optimal LP supply is \textbf{hump-shaped} in volatility--an intermediate volatility balances the cost from informed price correction and return from real noise traders and overrun arbitrages.

\section{Literature Review}
Our study contributes to the general understanding of AMM mechanisms in decentralized finance. Blockchain's general economic implications and analysis have been comprehensively reviewed by literature like~\cite{abadi2018blockchain} ~\cite{chen2021brief} ~\cite{harvey2021defi}, especially~\cite{lehar2025decentralized}, which collected all Uniswap's data since its launch and conduct a well-rounded empirical analysis of the market stability of the protocol. 

For more specific components of AMM, from the perspective of LPs~\cite{milionis2022automated} studies the cost incurred by adverse selection of informed traders to LPs and~\cite{aoyagi2020liquidity} studies the optimal liquidity provision between competing LPs.

Both LP and arbitrageur's problems were studied separately under various assumptions in the past literature. From the LP's perspective, ~\cite{goyal2023finding} assumed that LPs have a static belief about future asset prices, and established a model of optimal trading functions conditioning on the beliefs; whereas ~\cite{fan2021strategic} utilized neural networks to represent LPs' behavior across different price slots and analyzed the optimal strategies under the assumption transaction costs (e.g., slippage and transaction fees), ~\cite{park2021conceptual} argue that the AMM pricing rule facilitates inefficient trading, such as sandwich attacks, and may not be aligned with the optimality for LPs. On the other hand, arbitrageurs' optimal behaviors were also extensively studied; with the assumption of fixed transaction costs, ~\cite{angeris2022optimal} derives arbitrageur's optimal routing scheme when facing multi-asset and multi-pool market conditions; ~\cite{angeris2021note} proposed a CFMM privacy attack model assuming that arbitrageurs can access marginal prices, trade verification criteria, and reserve adjustments without incurring transaction fees and LP fees. 

Several studies are closely related to our work. ~\cite{aoyagi2021coexisting} analyzes a theoretical model of coexisting market structures of a limit-order book (LOB) and an AMM and investigate their interactions in terms of liquidity. They show that fluctuations in liquidity in the AMM can lead to a positive spillover effect on liquidity in the LOB market and the coexisting environment leads to a hump-shaped reaction of liquidity supply in the AMM to asset volatility; ~\cite{routledge2024automated} formulated and solved a static equilibrium model with LPs and informed and uninformed traders, pointed out that the optimal wealth allocation depends on the share of informed and uninformed traders in the market; where ~\cite{milionis2023automated} studies a stationary distribution of price ratios between the AMM price and an CEX price taking into account LP fees, and the LP's wealth dynamics are characterized with the same fashion as Loss-versus-rebalancing strategy in ~\cite{milionis2022automated}, but including LP fees and non-linear gas fees, which induces much richer and realistic dynamics.

Unlike prior works that analyze LPs and arbitrageurs in isolation or together but with simplified assumptions, this paper advances the existing literature by introducing a comprehensive dynamic equilibrium framework that simultaneously accounts for both LPs' strategies, LP-fees, and gas fee dynamics. 
this paper provides a  theoretical model of market equilibrium between heterogeneous participants in constant function market maker~(CFMM) automated market makers. 
Moreover, this framework is distinctively micro-founded through extensive real on-chain data from AMMs operating on CFMM mechanisms, allowing us to capture practical market dynamics that previous theoretical models have overlooked. This integration of empirical grounding with theoretical modeling offers novel insights into price efficiency, liquidity provision incentives, and the impact of transaction costs that significantly extend the current understanding of AMMs.

\section{Institutional Background and Data}

\subsection{Centralized Exchanges and AMM-DEX.}

\paragraph{Centralized Exchange (CEX).}
Assume that there is an observable centralized exchange with market price $P_t^A, P_t^B$, denominated in stablecoin (USDT, USDC, etc.) for assets $A$ and $B$ that follow independent GBM where the arbitrageur can trade assets $A$ and $B$ freely at the price $P^A_t, P^B_t$ at time $t$ without a transaction fee, also we assume that arbitrageurs anchor their valuations of cryptocurrencies to the CEX prices, i.e., economically, the supply and demand clears at CEX and AMMs operate on a CFMM constrant; according to a report by Bank for International Settlements ~\cite{bic2023crypto}, CEXs like Binance and Coinbase aggregate ~75\% of trading volume for major assets (BTC, ETH), making their prices the default reference for most investors; our choice of centralized exchange (CEX) prices as the anchor for price discovery is supported by empirical evidence from market microstructure analysis. Using high-frequency data from Ethereum mainnet and major CEXs, we conduct Granger causality tests between Binance ETH/USDT prices and on-chain prices. 
The results show strong evidence that CEX prices Granger-cause on-chain prices , while the reverse causality is not statistically significant.  These findings align with the market's microstructure where arbitrageurs, constrained by transaction confirmation times, typically observe price discrepancies from CEX feeds before executing on-chain trades. This empirical evidence supports our modeling choice of using CEX prices as the primary reference for arbitrageur behavior and market equilibrium dynamics. 

\paragraph{Automated Market Maker (AMM)}
The decentralized exchange where arbitrageurs can trade AMMs with Liquidity Pools such that the reserve for assets A and B are $R^A_t$ and $R^A_t$, which satisfies the constant function market-making rule, such that the feasible set of reserves $\mathcal{C}$ is defined as
$$
\mathcal{C} \left\{(R^A_t, R^B_t) \in \mathbb{R}_{+}^2 : \quad f(R^A_t, R^B_t)=k\right\}
$$
where $f: \mathbb{R}_{+}^2 \rightarrow \mathbb{R}$ is referred to as the bonding function or invariant, and $k \in \mathbb{R}$ is a constant. In other words, the feasible set is a level set of the bonding function. The pool is defined by a smart contract that allows an agent to transition the pool reserves from the current state $\left(R^A_t, R^B_t\right) \in \mathcal{C}$ to any other point $\left(R_t^{A'}, R^{B'}_t\right) \in \mathcal{C}$ in the feasible set, so long as the agent contributes the difference $\left(R^{A'}_t-R^A_t, R^{B'}_t-R^B_t\right) = \left(\Delta A, \Delta B\right)$ to the pool. 
This paper specifically focuses on Uniswap V2 pool, therefore the bonding function is given by $f(x, y) = x y$. Denote the pool price as $Q_t = \frac{R^B_t}{R^A_t}$.
Additionally, trading in AMM incurs a fixed proportional fee $\gamma$ that is paid to LP, the rule follows the same setting as ~\cite{milionis2022automated} ~\cite{adams2020uniswap} ~\cite{adams2021uniswapv3},

\subsection{Data}
\label{subsec:data}

We combine on-chain transaction data from major AMM-based decentralized exchanges with off-chain price data from centralized exchanges (CEXs) to study arbitrage behavior, noise trading, and liquidity provision incentives.

\paragraph{On-chain AMM data.}
We collect transaction-level data from six large constant-product AMM pools across multiple token pairs on Uniswap and Binance-AMM, including ETH/USDT and BNB/USDT. 
For each swap, we observe traded quantities of both assets, pre- and post-trade pool reserves, LP fee parameters, transaction timestamps, and gas usage and gas prices.

\paragraph{Centralized Exchange prices.}
For each AMM pool, we obtain high-frequency price data from the most liquid corresponding CEX trading pair.
CEX prices are from major exchanges including Binance, Bybit, Bitget, and OKX, which serve as the external reference price for arbitrage and are synchronized to on-chain timestamps.

\section{Baseline Model}
\label{sec:baseline}
In this section, we will introduce our baseline model formulations and solve for the equilibrium with liquidity provider (LP) and representative arbitrageurs. 

\paragraph{Market and Asset Price Dynamics}
Fix a filtered probability space that satisfies the usual conditions $\left(\Omega, \mathcal{F},\left\{\mathcal{F}_t\right\}_{t \geq 0}\right)$.  
Consider two risky assets A and B in the continuous-time setting such that $t \in \mathbb{R}_{+}$, and we consider two kinds of market where arbitrageurs can trade assets, the AMM and CEX.

We assume the price changes follows geometric Brownian motion
$$
\begin{aligned}
& d P^A_t=\mu_A P^A_t d t+\sigma_A P^A_t d W_t \\
& d P^B_t=\mu_B P^B_t d t+\sigma_B P^B_t d W_t
\end{aligned}
$$ 
To simplify the calculation, since we have assumed that both market price processes follow independent GBMs, we can apply the \textit{change of numéraire} trick to the two assets, and use Ito's quotient rule to treat, w.l.o.g., asset $B$ as numéraire, and we define the relative price of asset A to asset B as $P_t = \frac{P^A_t}{P^B_t}$, such that we denote the relative price dynamic $\frac{dP_t}{P_t}$ by
$$
\begin{aligned} \frac{d P_t}{P_t} & =\frac{d \left(P^A_t / P^B_t\right)}{\left(P^A_t / P^B_t\right)}=\left(\mu_A-\mu_B+\sigma_B\left(\sigma_B-\sigma_B\right)\right) d t+\left(\sigma_A-\sigma_B\right) d B_t \\ 
& =\mu d t+\sigma d B_t
\end{aligned}
$$

\paragraph{LP Fee Structure with Numeraire.} We denote the fee in units of log price by $\gamma$ as in~\cite{milionis2023automated}. In particular, when the agent purchases the asset $A$ from the pool amd sell assets to the pool, the total fee charged is
$$
\left(e^{+\gamma}-1\right)|\Delta B| \qquad \left(1-e^{-\gamma}\right)|\Delta B|
$$
respectively.
Throughout the paper, we will assume that the fee pays out to LPs and do not enter the pool, thus do not change the invariant $K_t = R^A_t R^B_t$ before and after the trade. 
This treatment is different from Uniswap V2 implementation where fees are compounded to the pool, for an explanation of why those two mechanism are the same in economic terms, see Remark~\ref{remark:why_paid_out}.

\subsection{Arbitrageur's Optimization Problem}

\paragraph{Arbitrageur's Problem.}
With the market setting above, we are ready to state the first core problem of this paper, the arbitrageur's optimization problem. 
Suppose that the trade is completed in infinitesimal time, and denote the amount of $A_t, B_t$ after trade as $A'_t, B'_t$. Firstly, we state the static objective function for the arbitrageur
\begin{equation}
\begin{aligned}
\max W^{A}_t &\quad \text{s.t. }  \\
 & A'_t=A_t- \Delta A \geq 0 \\ 
 & B'_t=B_t-\Delta B(\Delta A) \geq 0 \\
 & (R^A_t + \Delta A)(R^B_t + \Delta B) = R^A_t R^B_t 
 \end{aligned}
\end{equation}
The first two constraints are solvency constraint for the arbitrageur, and the third is the constraint for AMM pool.

The AMM relative price of asset $A$ to $B$ is determined by the reserve of the pool, which is given by $Q_t = \frac{R^B_t}{R^A_t}$, and since the LP fee is included we can characterize the action region for the arbitrageurs to trade.  (1) If $P_t>Q^{t-} e^{\gamma}$, there is an opportunity to buy asset $A$ (swap $B$ for $A$) from the pool and sell in the outside centralized exchange, and (2) If $P_t < Q^{t-} e^{-\gamma}$, there is an opportunity to sell asset $A$ (swap $A$ for $B$) in the AMM and buy from the centralized exchange. 

Though similar results of action and inaction regions for arbitrageurs have been explored in ~\cite{milionis2023automated}, slippage is not taken into account due to their goal of computing a stationary equilibrium of the CEX-AMM relative price process, and they implicitly assume that traders have infinite wealth and the pool is infinitely deep; when the assumption is dropped, we have the following proposition.

\begin{proposition}
\label{prop:optimal_swap}
Assume that the reserve pool is always enough for any transaction and neither asset can be traded until the reserve is zero ($R^A_t > 0 \bigcap R^B_t > 0$). Suppose that the arbitrageurs only aims to maximize the instantaneous return, then the optimal swapping of asset $A$ and $B$ under case 1 and case 2 is given below
For \textbf{case 1}, the optimal decision is given by 
$$
\Delta A_{b}^*=\max \left\{\Delta A_{\max }, \Delta A_{\text{aff}}\right\}
$$
where $\Delta A_{\max }=\frac{R^B_{t-} e^{\gamma}}{P_t} - R^A_{t-}$ and $\Delta A_{\text{aff}}=-\frac{B_tR^A_{t-}}{R^B_{t-} e^{\gamma}+B_t}$. \\
For \textbf{case 2}, the optimal decision is given by 
$$
\Delta A_{s}^*= \min\{\sqrt{\frac{k e^{-\gamma}}{P_t}}-R^A_{t-}, A_t\}
$$ 
\end{proposition}
In both cases, the arbitrageur trades in all their depreciating assets until the market price converges, which partly coincides with the mispricing-determined action-inaction region given in~\cite{milionis2023automated}, When buying, the optimal behaviour depends on (1) the holding of numeraire and (2) the reserve of the pool; whereas when selling, the optimal behaviour depends on the (1) reserve for the risky asset and (2) the current holding of the risky asset.  
The full proof of this proposition can be found in Appendix~\ref{app:optimal_swap_baseline}.

Moreover, when trading with CFMM AMMs, there is an intrinsic asymmetry comes from the price impact of buying and selling resulting from the reserve proportion of the pool, and the proposition below gives the conditions of where asymmetry arises from buying and selling at the CFMM when $f(x,y) = xy$. 

\begin{proposition}
\label{prop:b_s_asymmetry}
Given the reserve for two assets in the pool are $R_B, R_A$, and absolute price impact as $|P - P'|$, where $P'$ is the price after the trade occurs, then from trader's perspective, buying induces more slippage than selling, e.g. for same quantity of buying or selling from the traders, buying always cause higher $|P - P'|$.
\end{proposition}

For full proof of the proposition see Appendix~\ref{app:proof_for_b_s_asymmetry}, and see Figure~\ref{fig:ETH_asymmetry}--\ref{fig:BNB_asymmetry} for empirical confirmation on different CFMM pools.

\subsection{LP's Dynamic Optimization Problem}
The reason we consider a dynamic problem for LPs are two folds, firstly, the solution to an LP's static problem is trivial due to impermanent loss which has been discussed extensively in past literature and research reports~\cite{milionis2022automated, loesch2021impermanentlossuniswapv3} ; secondly, it is more realistic to consider LP's problem for a longer period with the presence of one-time deposit handling fee. 
Therefore, we formulate a dynamic optimization problem for LPs.
We assume that LPs allocate their wealth of risky asset $A$ between AMM liquidity provision and direct holding, which we will denote liquidyity provison by $R_A$ and direct holding by $A^{out}$.
Before delving into the main part of the proof, we first solve for LP's wealth dynamic and then discuss LP's optimal wealth allocation under different utility forms. Firstly, we define the log variables
\[
p_t:=\log P_t,\qquad q_t:=\log Q_t,\qquad k_t:=\log K_t,
\]
where $P_t$ is the outside (CEX) relative price (num\'eraire $B$) and $Q_t=R_{B,t}/R_{A,t}$ is the AMM marginal price.
The CEX price follows a GBM:
\[
\frac{dP_t}{P_t}=\mu\,dt+\sigma\,dB_t
\quad\Longleftrightarrow\quad
dp_t=\Big(\mu-\tfrac{1}{2}\sigma^2\Big)dt+\sigma\,dB_t.
\]

\paragraph{Constant product in logs.}
Under the CFMM invariant $R_{A,t}R_{B,t}=K_t$ and $Q_t=R_{B,t}/R_{A,t}$, the reserves are
\[
R_{A,t}=\exp\!\Big(\tfrac{1}{2}(k_t-q_t)\Big),\qquad
R_{B,t}=\exp\!\Big(\tfrac{1}{2}(k_t+q_t)\Big).
\]

\paragraph{No-arbitrage band and asymmetric correction arrivals.}
Let $X_t:=p_t-q_t=\log(P_t/Q_t)$. The $\gamma$-fee implies a no-arbitrage region $X_t\in[-\gamma,\gamma]$.
We model boundary corrections by two counting processes $(N_t^-,N_t^+)$ with intensities $(\lambda_t^-,\lambda_t^+)$, we will also refer to $\lambda^{\pm}$ as adverse selection intensity throughout the paper:
$dN_t^-=1$ represents a buy-side correction (arbitrageurs buy $A$ from the AMM) and $dN_t^+=1$ a sell-side correction.
After a correction,
\[
q_t = p_t-\gamma \ \ \text{if } dN_t^-=1,
\qquad
q_t = p_t+\gamma \ \ \text{if } dN_t^+=1.
\]

\paragraph{Borrowing account.}
Let $M_t=\exp(\int_0^t r_s^{*}ds)$ be the exogenous money-market account, $dM_t=r_t^{*}M_tdt$.
Let $D_t$ denote LP debt, evolving as
\[
dD_t=r_t^{*}D_tdt + dR^{LP}_{B,t},
\]
where $dR^{LP}_{B,t}$ is the num\'eraire flow injected/withdrawn at LP liquidity update times.

\paragraph{Wealth process.}
Define LP wealth in num\'eraire units:
\[
W_t := A^{out}_t P_t + R_{A,t}P_t + R_{B,t} - D_t,
\]

\begin{proposition}[LP Wealth Dynamics]
\label{prop:log_wealth_dynamics}
The wealth process satisfies the jump--diffusion decomposition
\begin{equation}
\label{eq:wealth_dynamics}
dW_t
=
(A^{out}_t+R_{A,t-})\,dP_t
-r_t^{*}D_t\,dt
+\Delta W_t^-\,dN_t^- +\Delta W_t^+\,dN_t^+,
\end{equation}
where the arbitrage-correction jump sizes are
\[
\Delta W_t^-=
2\exp\Big(\tfrac{k_{t-}+p_t+\gamma}{2}\Big)
-\Big(P_tR_{A,t-}+e^\gamma R_{B,t-}\Big),
\]
\[
\Delta W_t^+=
2\exp\Big(\tfrac{k_{t-}+p_t-\gamma}{2}\Big)
-\Big(P_tR_{A,t-}+e^{-\gamma} R_{B,t-}\Big).
\]
Equivalently, using compensated processes $d\widetilde N_t^\pm:=dN_t^\pm-\lambda_t^\pm dt$,
\begin{align*}
dW_t
&=
(A^{out}_t+R_{A,t-})\mu P_t\,dt
+(A^{out}_t+R_{A,t-})\sigma P_t\,dB_t
-r_t^{*}D_t\,dt \\
&\quad
+(\lambda_t^-\Delta W_t^-+\lambda_t^+\Delta W_t^+)\,dt
+\Delta W_t^-\,d\widetilde N_t^-+\Delta W_t^+\,d\widetilde N_t^+.
\end{align*}
\end{proposition}

\subsection{Liquidity Provision with CFMM Correction jumps}
\label{subsec:crra_hjb_cfmm_jumps}

Let $Q_{t-}$ denote the AMM price immediately before a correction and define the pre-correction mispricing ratio
\[
m \;:=\; \frac{P}{Q_{-}}.
\]
Let $M^{-}$ (resp. $M^{+}$) denote the random variable $m$ conditional on a buy-side correction (resp. sell-side correction),
so that $M^{-}>e^{\gamma}$ and $M^{+}<e^{-\gamma}$ almost surely.

\begin{lemma}[Per-event proportional wealth return under fee-extracted CFMM correction]
\label{lem:Jpm}
Consider a constant-product pool with pre-correction reserves $(R_A,R_B)$ and pre-correction price $Q_-:=R_B/R_A$,
and suppose LP fees are \emph{paid out} and not compounded into reserves, so that the invariant $K:=R_AR_B$ is unchanged by swaps.
Fix CEX price $P$ and define $m:=P/Q_-$.

\begin{enumerate}
\item (Buy-side correction; $m>e^{\gamma}$ a.s. ) If arbitrage trading pushes the AMM to the fee-adjusted boundary
$Q'=P e^{-\gamma}$, then the total (mark-to-market plus paid-out fee) value of the LP's AMM position, measured at price $P$,
jumps by the multiplicative factor $1+J^{-}(m;\gamma)$ where
\[
J^{-}(m;\gamma)
:=\frac{2e^{\gamma/2}\sqrt{m}-m-e^{\gamma}}{m+1}
= -\frac{\big(\sqrt{m}-e^{\gamma/2}\big)^{2}}{m+1}
\;\;<\;0.
\]

\item (Sell-side correction; $m<e^{-\gamma}$ a.s.) If arbitrage trading pushes the AMM to the fee-adjusted boundary
$Q'=P e^{+\gamma}$, then the total (mark-to-market plus paid-out fee) value of the LP's AMM position, measured at price $P$,
jumps by the multiplicative factor $1+J^{+}(m;\gamma)$ where
\[
J^{+}(m;\gamma)
:=\frac{2e^{-\gamma/2}\sqrt{m}-m-e^{-\gamma}}{m+1}
= -\frac{\big(\sqrt{m}-e^{-\gamma/2}\big)^{2}}{m+1}
\;\;<\;0.
\]
\end{enumerate}
\end{lemma}

\paragraph{Wealth dynamics with reduced-form jump representation.}
Let $N_t^{-}$ and $N_t^{+}$ be counting processes with intensities $\lambda^{-}$ and $\lambda^{+}$.
At buy-side (resp. sell-side) corrections, the LP's AMM-deployed wealth experiences proportional jump returns
$J^{-}(M^{-};\gamma)$ (resp. $J^{+}(M^{+};\gamma)$) from Lemma~\ref{lem:Jpm}.
Write the controlled wealth process as Equation~\ref{eq:wealth_dynamics}. We have the following theorem

\begin{theorem}[Liquidity Provision under Purely Informed Trading]
\label{thm:crra_hjb_cfmm}
Fix $\rho>0$ and CRRA utility $u(w)=\frac{w^{1-\eta}}{1-\eta}$ with $\eta>0$, $\eta\neq 1$.
Consider an LP who chooses the AMM allocation $\{\theta_t\}_{t\ge0}$ to solve
\[
V(w)=\sup_{\{\theta_t\}}\ \mathbb{E}\!\left[\int_0^\infty e^{-\rho t}\frac{W_t^{1-\eta}}{1-\eta}\,dt\ \Big|\ W_0=w\right],
\]
subject to the wealth dynamics Equation~\ref{eq:wealth_dynamics} and portfolio constraints
$\theta_t\in[\theta_{\min},\theta_{\max}]$.

Assume that all order flow against the AMM is generated by informed arbitrageurs,
so that LPs face only adverse-selection risk and no exogenous noise trading.
Let the admissible set be
\[
\Theta_{\mathrm{adm}}
:=
\Big\{\theta\in[\theta_{\min},\theta_{\max}]:
1+\theta J^{-}(M^{-};\gamma)>0\ \text{a.s.},\ 
1+\theta J^{+}(M^{+};\gamma)>0\ \text{a.s.}
\Big\},
\]
and define the reduced objective
\begin{align}
\label{eq:reduce_objective}
\Phi(\theta)
&:=
\mu-r^{*}\theta-\frac{\eta}{2}\sigma^2
+\frac{\lambda^{-}}{1-\eta}\,
\mathbb{E}\!\left[(1+\theta J^{-}(M^{-};\gamma))^{1-\eta}-1\right]
+\frac{\lambda^{+}}{1-\eta}\,
\mathbb{E}\!\left[(1+\theta J^{+}(M^{+};\gamma))^{1-\eta}-1\right].
\end{align}

Then the following results hold:
\begin{enumerate}
\item (\emph{Weak Dominance of Minimal Liquidity.})
The function $\Phi$ is strictly concave on $\Theta_{\mathrm{adm}}$.
Its unique maximizer satisfies
\[
\theta^{*}=\theta_{\min}.
\]
That is, when all traders are informed arbitrageurs, allocating additional
capital to the AMM is weakly dominated by direct holding of the risky asset.

\item (\emph{Knife-edge Interior Condition.})
An interior optimum $\theta^{*}\in\mathrm{int}(\Theta_{\mathrm{adm}})$
can occur only if $\theta^{*}$
satisfies the first-order condition
\begin{align}
\label{eq:FOC_baseline}
r^{*}
&=
\lambda^{-}\,
\mathbb{E}\!\left[J^{-}(M^{-};\gamma)\,(1+\theta^{*} J^{-}(M^{-};\gamma))^{-\eta}\right]
+\lambda^{+}\,
\mathbb{E}\!\left[J^{+}(M^{+};\gamma)\,(1+\theta^{*} J^{+}(M^{+};\gamma))^{-\eta}\right].
\end{align}
Such an interior solution is non-generic and disappears under arbitrarily
small increases in adverse selection intensity $\lambda^{\pm}$.

\item (\emph{Value Function.})
The value function takes the CRRA form
\[
V(w)=\frac{C}{1-\eta}\,w^{1-\eta},
\qquad
C=\frac{1}{\rho-(1-\eta)\,\Phi(\theta^{*})},
\]
provided the transversality condition $\rho>(1-\eta)\Phi(\theta^{*})$ holds.
\end{enumerate}
\end{theorem}

The theorem formalizes a benchmark that when all trades against the AMM are generated by informed arbitrageurs, liquidity provision exposes the LP exclusively to adverse selection:
arbitrageurs trade precisely when the AMM price is stale relative to the external market, while fee income is insufficient to compensate for the resulting impermanent loss and financing cost.
As a consequence, the LP optimally minimizes exposure to the AMM,
allocating only the lowest admissible level of liquidity.
In this sense, an AMM populated solely by informed traders does not sustain economically meaningful liquidity provision in equilibrium.

This theoretical prediction stands in sharp contrast with observed AMM markets, which exhibit substantial and persistent liquidity even during periods of elevated volatility.
The discrepancy suggests that real-world order flow cannot be explained by informed arbitrage alone.
In the next section, we provide empirical evidence that a large fraction of AMM volume is unprofitable ex-post and co-moves strongly with external volatility, motivating a model with heterogeneous traders in which noise trading and failed arbitrage attempts play a central role in sustaining liquidity.

\section{Empirical Evidence: Noise Trading, Volatility, and Gas}
\label{sec:empirical}

\subsection{Data Preprocessing}
\paragraph{Time window and aggregation.}
The dataset spans multiple market regimes, including both low- and high-volatility periods.
To study state dependence, transactions are aggregated into fixed-length time intervals when computing volatility, trading volume, and gas-related metrics. 
For market with lower transaction frequency, e.g.~ETH/USDT pool, we aggregate data in 10-30mins windows; whereas for other more active market~POL/USDT pool, the aggregation window is 3mins.

\paragraph{Trade classification.}
For each on-chain swap, we compute the ex-post payoff relative to the contemporaneous CEX price, net of LP fees implied by the AMM pricing rule and gas costs.
The classification is outcome-based and does not rely on assumptions about traders’ ex-ante intentions.
We analyzed on-chain transactions from major AMMs and identified a significant amount of noise trading behaviors. For data pre-processing, we distinguish profitable and unprofitable transactions from the same set of real on-chain data by comparing the on-chain price and the transaction-fee-adjusted CEX prices. The on-chain price is calculated by reserve change in the pool, which incorporates LP fees already, e.g.~for ETH/USDT pool, the on-chain price is  
$$
Q_t = \left|\frac{\Delta USDT}{\Delta ETH}\right|
$$
and the transaction fee-adjusted CEX price is given by 
$$
P_t = \begin{cases}
\max(P_t^i) \times (1 - e^{-\gamma})  \text{ if } \Delta ETH < 0 \\
\min (P_t^i) \times (e^{\gamma}-1) \text{ if } \Delta ETH > 0
\end{cases}
$$
Here we use 0.1\% for Binance and Bitget.  Since we consider four major CEXs, including ByBit, Bitget, Binance, and OKX, the final $P_t$ is a combined, transaction-fee-adjusted price. 
When $\Delta ETH < 0$, we say the traders take a long position~(because they are buying from AMM) and a short position if $\Delta ETH > 0$. This $\Delta ETH$ is from the perspective of AMM, which is consistent with our previous formulation. 

Based on our definition of arbitrageurs, we filter out the profitable transactions, such that a transaction is profitable if and only if
$$
\begin{cases}
\Delta A(Q_t - P_t e^{\gamma}) - g_t > 0 \text{ if } \Delta ETH > 0 \\
\Delta A( P_t e^{-\gamma} - Q_t) - g_t > 0 \text{ if } \Delta ETH < 0
\end{cases}
$$
where $g_t=g(v_t)$ is the gas fee. The metrics we defined can be found in 

\paragraph{Descriptive variables.}
All descriptive measures used in the analysis are defined as follows.
The \emph{price gap} is the relative deviation between the AMM price and the CEX price immediately before a transaction.
The \emph{profitability rate} is the fraction of trades within a given time window or volatility bucket that are profitable.
\emph{Volatility} is measured as realized volatility of the CEX price over rolling windows and discretized into volatility buckets.
\emph{Noise trading volume} is defined as the total traded volume of ex-post unprofitable trades within a time interval.
\emph{Gas metrics} include the average gas price and total gas expenditure per aggregation window, capturing network congestion and priority costs.

\subsection{Buy and Sell Asymmetry}
Examining on-chain data across multiple asset pairs, we document a pronounced asymmetry between buying and selling activity on AMMs during periods in which the CEX price exhibits no directional movement—i.e., when prices are locally stable or approximately mean-reverting. In these drift-neutral regimes, \emph{sell-side trades} (e.g., ETH
→
→USDT) account for both a larger share of profitable transactions and a greater fraction of total profitable volume than buy-side trades.

This pattern is consistent with the price-impact asymmetry implied by the constant-product geometry and formalized in Proposition~\ref{prop:b_s_asymmetry}. When prices fluctuate within a band, profitable opportunities are scarce and small. Because buying from the AMM moves the pool toward the scarce reserve and induces steeper marginal price impact, buy-side trades require larger mispricing excursions to remain profitable net of slippage, fees, and execution costs. As a result, buy-side trades are disproportionately filtered out when conditioning on realized profitability. In contrast, sell-side trades face milder effective price impact for comparable notional adjustments, allowing a larger fraction of such trades to remain profitable in low-volatility environments.

Figure~\ref{fig:ETH_asymmetry} illustrates this mechanism for the ETH/USDT pool. Across volatility regimes with negligible price drift, ETH->USDT trades exhibit a consistently higher profitability rate than USDT->ETH trades. Because directional price movements are absent by construction, this asymmetry cannot be attributed to trends or inventory rebalancing motives. Instead, it reflects a selection effect induced by the asymmetric curvature of the AMM pricing function: conditional on small mispricing shocks, sell-side trades are more likely to survive the fee- and slippage-adjusted profitability threshold. 
Similar results obtain across other markets. Figure~\ref{fig:BNB_asymmetry} reports analogous patterns for the BNB/USDT pool on the BNB chain, indicating that the buy–sell asymmetry is a robust feature of Uniswap V2–style constant-product AMMs rather than an artifact of asset-specific dynamics or transient price movements.

\subsection{Profitable v.s. Unprofitable Trades}
\label{subsec:profitable_unprofit}

\paragraph{Trader Heterogeneity}
The second key empirical finding is the strong evidence of heterogeneous traders.
Figure~\ref{fig:ETH_Noise} reveals a striking disparity between profitable and unprofitable trades that strongly supports the hypothesis of heterogeneous traders. 
Profitable trades likely representing arbitrageur transactions constitute only 10\% of all transactions (1,913 out of 19,069 total trades) but account for approximately 55.68\% of the total transaction value (\$33.17M out of \$59.57M). 
From Figure~\ref{fig:ETH_asymmetry}, we can see that as the market volatility increases, the frequency of profitable arbitrage rises substantially but still still makes up for a small portion of all transactions.

The volume distribution shows clear separation between the two groups: profitable trades have substantially higher average transaction volumes, clustering predominantly in the \$10,000-\$100,000 range, where unprofitable trades are much more numerous (17,156 trades) but have dramatically smaller average transaction sizes (\$1,539 mean, \$249 median), with most concentrated between \$100-\$10,000.

This pattern closely mirrors traditional financial markets where informed investors~(often with superior information and larger capital bases) generate most returns through larger volume trades, while retail/uninformed traders execute more frequent but smaller and less profitable transactions, as recorded by~\citet{duong2009order} and~\citet{barardehi2022institutional}. 
The log-scale density plot visually confirms this bifurcation, showing the profitable trade distribution shifted significantly rightward toward higher volumes; Besides the ETH/USDT pool, other similar AMM V2 pools all exhibit this disparity between profitable and unprofitable transactions. More phenomenon can be detected in other pools with major asset/USDT pair, see Figure~\ref{fig:BNB_Noise} and~\ref{fig:POL_Noise}. 

\paragraph{Comovement between CEX Price Volatility and Unprofitable Trading Volume}
Another important empirical finding is that unprofitable trading volume in AMMs systematically co-moves with CEX price volatility.
Table~\ref{tab:unprofitable_vol_comove} reports Granger causality tests showing that volatility in centralized exchanges significantly predicts subsequent unprofitable trading volume on AMMs across all three asset pairs, with particularly strong and robust effects for ETH/USDT and POL/USDT. The predictive relationship persists across multiple lags, indicating that periods of elevated market volatility are followed by an increase in unprofitable on-chain trades.

This pattern is difficult to reconcile with a brute-force classification that interprets all unprofitable transactions as originating from exogenous noise traders, since genuine noise trading—driven by liquidity needs, operational frictions, or non-informational motives—should not systematically intensify in response to external price volatility. Instead, the evidence suggests that a substantial fraction of unprofitable trades are endogenously generated by the trading mechanism itself.

Motivated by this observation, in Section~\ref{subsec:race_game} we microfound this phenomenon through an informed arbitrageur race game. In the model, multiple arbitrageurs compete to exploit transient price discrepancies between the AMM and centralized exchanges. While early movers earn positive profits, late-arriving arbitrageurs may have their orders executed after the price has already been corrected, leading to negative realized profits. This mechanism provides a structural explanation for why unprofitable trading volume increases precisely during high-volatility periods, without attributing it solely to noise traders.

\subsection{Gas Fee Dynamics}
The third key focus of our empirical analysis is modelling the dynamics of gas fees. Gas fees serve as a proxy for transaction congestion on the blockchain, significantly influencing the behavior of both traders and LPs. Recognizing the critical role of gas fees in determining optimal trader behavior, we explicitly endogenize these fees into our current AMM model. To achieve this, we perform causal and regression analyses using gas prices and readily accessible on-chain transaction data. 

A natural choice to model gas prices is as a function of the price gap between AMM and one of the major CEXs. Under the rationality assumption, when the price gap exceeds the LP-fee-adjusted price, arbitrageurs will trade with AMM until the prices converge. This activity, in theory, should drive on-chain traffic congestion, thereby increasing gas fees. However, our empirical analysis of actual data challenges this assumption: 
both the price gap and the price gap ratio\footnote{Calculated as $\frac{\text{price gap}}{\text{CEX price}}$.} exhibit minimal explanatory power concerning gas price variations. 
Furthermore, the Granger causality test examining the relationship from price gap and price gap ratio to gas prices yields statistically insignificant p-values.

Table \ref{tab:regression_and_causality} provides a comprehensive analysis of the determinants of gas prices, revealing notable insights into trader behavior. 
Our results indicate that, compared to price gaps alone, volatility—calculated using historical rolling windows of CEX prices—exhibits significant statistical causality across all considered window length. Additionally, regression analysis confirms that CEX price volatility possesses robust explanatory power regarding variations in gas prices.

Together, the three facts motivate an extended equilibrium framework in which liquidity provision and trading outcomes are shaped not only by informed arbitrage, but also by noise trading, arbitrage races, endogenous gas fees, and time-varying volatility.

\section{Extended Model: Gas Fees, Stochastic Volatility, Persistent Noise Fees, and Arbitrage Races}
\label{sec:extended}

Motivated by the empirical evidence from Section~\ref{sec:empirical}, this section extends the baseline by incorporating four empirically motivated components:
\begin{enumerate}
\item \textbf{Time-varying volatility} governed by a mean-reverting process.
\item \textbf{Gas fees} that respond to congestion and shape an arbitrage \emph{no-arbitrage band};
\item \textbf{Real noise trading} driven by idiosyncratic shock rather than instantaneous profit maximization;
\item \textbf{Overrun arbitrage order flow} driven by arbitrageurs who aim to maximize instantaneous profit but failed to executed the order on time.
\end{enumerate}

The key modeling discipline is to preserve the baseline \emph{equilibrium mapping} that LP chooses $\theta$, arbitrageurs exploit price gap whenever profitable, while allowing (i) non-arbitrage order flow to generate fee revenue, (ii) congestion-sensitive gas costs to create a state-dependent inaction region, and (iii) stochastic volatility to scale both the arrival of profitable opportunities and the intensity of failed overrun arbitrage attempts. 

\paragraph{Gas Fee}
We model gas fee as a function of outside liquidity, we do not impose specific functional form but $g_t = g(v_t)$ such that function $g$ is increasing, continuous, and convex on $[0, \infty]$.

\paragraph{CEX price with stochastic variance.}
We keep the num\'eraire $B$ and model the outside (CEX) relative price $P_t$ under stochastic volatility:
\begin{equation}
\label{eq:sv_price}
\frac{dP_t}{P_t}=\mu\,dt+\sqrt{v_t}\,dB_t
\end{equation}
The variance follows a CIR/Heston-type mean-reverting process
\begin{equation}
\label{eq:sv_var}
dv_t=\kappa(\bar v-v_t)\,dt+\xi\sqrt{v_t}\,dW_t,
\qquad
d\langle B,W\rangle_t=0,
\end{equation}
where $\kappa>0$ is the mean-reversion speed, $\bar v>0$ is the long-run variance, and $\xi>0$ is the volatility-of-variance.

\subsection{Real Noise Traders}
\label{subsec:noise_symmetric_fee}

We assume there exists a continuum of \emph{real} noise traders who demand token $A$ for exogenous, idiosyncratic reasons
(inventory/payment needs, privacy, composability). Their trades are \emph{not} driven by the CEX--AMM mispricing; instead, they depend only on
pool depth and an idiosyncratic preference shock.

\paragraph{Depth proxy.}
Let
\begin{equation}
\label{eq:depth_proxy_sym}
(R_{A,t}R_{B,t})^{\beta}=K_t^{\beta},\qquad \beta>0,
\end{equation}
denote the pool-depth proxy, increasing in AMM liquidity, $\beta$ here captures the effective scaling of overrun arbitrageur and noise-trader behavior with pool size, which is set to be calibrated in future empirical works.

\paragraph{Idiosyncratic shock.}
At each time $t$, a representative noise trader draws an i.i.d.\ preference shock
\begin{equation}
\label{eq:xi_normal}
\xi_{N,t}\sim \mathcal N(0,\sigma_N^2),
\end{equation}
independent of the CEX Brownian motions and of the arbitrage correction processes.

\begin{lemma}[Real Noise Trader Fee Generation]
\label{lem:noise_trader_fee}
A representative noise trader chooses a signed trade size $q\in\mathbb R$ in asset $A$ (positive means net buy of $A$ from the AMM)
by solving the local quadratic problem
\begin{equation}
\label{eq:noise_problem_sym}
\max _{q \in \mathbb{R}} \xi_{N, t} q-\frac{1}{2} \frac{q^2}{K_t^\beta}-g\left(v_t\right) .
\end{equation}
Then the expected instantaneous paid out fee is given by
\begin{equation}
\label{eq:noise_trading_fee}
\mathbb{E}\left[\dot{F}_t^N \mid \mathcal{F}_{t-}\right]=Q_t \cdot\left(e^\gamma-e^{-\gamma}\right) \cdot K_t^\beta \sigma_N \varphi\left(\frac{1}{\sigma_N} \sqrt{\frac{2 g\left(v_t\right)}{K_t^\beta}}\right) . 
\end{equation}
where $\varphi$ is the pdf for standard normal distribution. Moreover, for any fixed $K>0$, if $g'(v)>0$ then
\begin{equation}
\label{eq:noise_fee_v_derivative}
\frac{\partial}{\partial v}\,
\mathbb E[\dot F_t^{N}\mid K_t=K, Q_t]
=
-\,Q_t\cdot (e^\gamma-e^{-\gamma})\cdot \frac{g'(v)}{\sigma_N}\,\varphi\!\big(z(v,K)\big)
\;<\;0.
\end{equation}
If in addition $g(v)\to\infty$ as $v\to\infty$, then $\varphi(z(v,K))\to 0$ and hence
$\mathbb E[\dot F_t^{N}\mid K,Q]\to 0$ as $v\to\infty$.
\end{lemma}
The intuition is that real noise traders trade because of idiosyncratic needs rather than expected profit, so their desired trade size is pinned down by a private preference shock rather than by market prices. 
However, market-wide gas fees and pool depth act as participation and execution frictions: higher gas fees truncate small preference-driven trades, while deeper pools attenuate slippage and allow those trades to go through. 
As volatility raises gas fees, fewer noise trades are executed, so fee revenue from genuine noise traders declines and eventually vanishes at high volatility despite unchanged underlying preferences. See~\ref{app:proof_for_noise_trader} for proof.

\subsection{Arbitrageur's Optimization Problem}

Similar to noise traders, with gas fee, arbitrage entry problem becomes: conditional on trading, the optimal swap size coincides with the baseline optimizer since gas is a fixed entry cost, but the arbitrageur \emph{enters} only when the expected realized profit exceeds gas. Thus the optimal policy takes the form.
\[
\Delta A_{t,s}^{*}=
\begin{cases}
\Delta A_{t,s,\text{baseline}}^{*}, & \text{if }\pi_s\!\left(\Delta A_{t,s,\text{baseline}}^{*}\right)>0,\\
0,&\text{otherwise},
\end{cases}
\]
and for selling, the optimal policy is symmetric.
Thus, the intensity of jump event $\lambda^{\pm}$ induced by arbitrageurs now depends on gas fee and thus on the volatility $v_t$, which we write as $\lambda^{\pm}(g_t) =  \lambda^{\pm}(g(v_t))$. 
Moreover, we have that $\lambda^{\pm, \prime}(g(v_t)) < 0$, such that the

\subsection{Arbitrage Race and Overrun Order Flow}
\label{subsec:race_game}
Measured “noise” order flow in AMMs is a composite object. Our profitability-based classification labels as “noise” any trade whose realized PnL is negative relative to fee-adjusted CEX prices. 
However, this set mixes two economically distinct behaviors. Real noise traders transact for venue-specific reasons~(custody, composability, payment/inventory needs, or operational constraints) that are largely unrelated to short-horizon CEX volatility, implying a comparatively stable baseline flow. 
By contrast, periods of elevated volatility generate more frequent and larger transient CEX–AMM price dislocations, attracting a surge of arbitrage entry. 
Because on-chain execution is organized as a priority race, only the first transaction captures the spread; many profit-seeking arbitrage orders execute at a stale pool state or bid away the spread via priority fees, yielding negative ex-post PnL and thus being recorded as “unprofitable”. 
This mechanism explains why the unprofitable component of volume co-moves strongly with volatility as given in the empirical section~\ref{subsec:profitable_unprofit}. Therefore, we introduce the formulation below.

First, \textbf{real} noise traders trade for idiosyncratic reasons unrelated to mispricing (Section~\ref{subsec:noise_symmetric_fee}) and generate fee income that compensates LPs.
Second, we define \textbf{ex-post} unprofitable transactions made by \textbf{arbitrageurs} who submit profit-seeking transactions but are \textbf{overrun} in the priority race: their transactions execute after the AMM state has already been updated by a faster competitor, producing \emph{negative} realized PnL and appearing as ``unprofitable'' in ex-post profitability metrics, we will refer to these traders as overrun arbitrageurs throughout the rest of the paper.

We micro-found this channel via an \textbf{arbitrage race game} and we approximate their price impact by the local quadratic CFMM expansion, which is approximately symmetric across directions at this scale. 
For clarity we present the buy-side event; the sell-side is symmetric, moreover, we assume that overrun flow never push the pool price outside the no-arbitrage band.

Let $\Delta_t$ denote the pre-trade opportunity ``spread'' in price units observed by potential arbitrageurs at time $t$. We model
\[
\Delta_t \mid v_t \sim \mathcal N(0,v_t),
\qquad
|\Delta_t| \text{ is half-normal with scale } \sqrt{v_t}.
\]
Therefore we have 
\[
\mathbb{E}\left[\Delta_t \mid \Delta_t>0, v_t\right]=\sqrt{\frac{2}{\pi}} \sqrt{v_t}
\]
A single opportunity at time $t$ unfolds as:
(i) $\Delta_t$ is realized and observed up to negligible latency;
(ii) each arbitrageur chooses entry and, conditional on entry, a quantity and a priority bid (embedded as an effective per-transaction cost);
(iii) a builder orders transactions by priority, so the first execution updates the AMM and subsequent transactions fill at a stale state.

\begin{proposition}[Race-game Quantity Choice, Entry, and Overrun Volume scaling]
\label{prop:race_overrun}
Consider a buy-side opportunity at time $t$ with spread $\Delta_t>0$ and pool-depth proxy $K_t^{\beta}$, $\beta>0$.
A buy-side attempt of size $q\ge 0$ yields payoffs
\[
\Pi^{\mathrm{win}}(q)=q\Delta_t-\frac{q^2}{2K_t^{\beta}}-g_t,\qquad
\Pi^{\mathrm{lose}}(q)=-\frac{q^2}{2K_t^{\beta}}-g_t.
\]
Arbitrageur $i$ holds a subjective belief $\hat N_i:=\mathbb E[N\mid \mathrm{info}_i]$ about the number of entrants $N$, and we use the
reduced-form win probability $\pi_i\approx 1/\hat N_i$.
Conditional on entry, $i$ submits quantity $q_i\ge 0$ solving
\[
\max_{q_i\ge 0}\ \pi_i q_i\Delta_t-\frac{q_i^2}{2K_t^{\beta}}-g_t,
\]
the total overrun quantity is the sum of overrun arbitrageurs' submitted quantities and the expected overrun volume is scaled by:
\[
L_A=\sum_{i\neq w} q_i^*
=K_t^{\beta}\Delta_t \sum_{i\neq w}\frac{1}{\hat N_i}, \qquad
\mathbb{E}\left[L_A \mid K_t, v_t\right] \propto K_t^\beta \sqrt{v_t}
\]

\end{proposition}

\paragraph{Overrun ratios.}
To assist the later derivation of LP's optimal liquidity provision, we use a dimensionless transform to map the overrun volume to ratio.
Let $L^-$ (resp.\ $L^+$) denote the realized overrun order flow (in token $A$ units) executed \emph{after} the winner in a buy-side
(resp.\ sell-side) correction block. Define the dimensionless overrun ratios $U^{\pm}:=\frac{L^{\pm}}{R_A^{b,{\pm}}}\in[0,1)$. 
For buy-side, overrun arbitrageurs remove $A$ so $U^-<1$ ensures solvency of reserves; for sell-side, overrun arbitrageurs add $A$ so $U^+\ge 0$.
The expected overrun ratio scales like
\[
\mathbb{E}\left[U \mid K_t, v_t, P_t\right]=\mathbb{E}\left[\frac{L_A}{R_A^b}\right] \approx \frac{K_t^\beta \sqrt{v_t} }{\sqrt{K_t e^\gamma / P_t}} \propto K_t^{\beta-\frac{1}{2}} \sqrt{v_t} \sqrt{P_t} .
\]

Then in the similar fashion of Lemma~\ref{lem:Jpm}, we provide the proportional wealth return under overrun arbitrageur jumps.

\begin{lemma}[Per-event proportional wealth return under overrun arbitrageur overrun]
\label{lem:Jext}
Consider a constant-product pool with pre-correction reserves $(R_A,R_B)$ and pre-correction price $Q_-=R_B/R_A$.
Assume fees are \emph{paid out} (not compounded), so $K=R_AR_B$ is unchanged by swaps. Fix CEX price $P$ and define $m=P/Q_-$.

\begin{enumerate}
\item \textbf{Buy-side correction ($m>e^\gamma$).}
The winner clears to the fee-adjusted boundary $Q'=P e^{-\gamma}$, and overrun arbitrageurs additionally remove a fraction $u\in[0,1)$ of the boundary
$A$-reserve (i.e., $u=U^-$). Then the total value of the LP's AMM position, measured at price $P$ and including paid-out fees from the
winner \emph{and} overrun arbitrageurs, jumps by the multiplicative factor $1+J_{\mathrm{ext}}^{-}(m,u;\gamma)$ where
\begin{equation}
\label{eq:Jext_minus}
J_{\mathrm{ext}}^{-}(m,u;\gamma)
=
\frac{
e^{\gamma/2}\sqrt{m}\Big((1-u)+\frac{1}{1-u}\Big)\;-\;m\;-\;e^{\gamma}
}{m+1} = \frac{e^{\gamma / 2} \sqrt{m} \cdot \frac{u^2}{1-u}-\left(e^{\gamma / 2}-\sqrt{m}\right)^2}{m+1}
\end{equation}

\item \textbf{Sell-side correction ($m<e^{-\gamma}$).}
The winner clears to $Q'=P e^{+\gamma}$, and overrun arbitrageurs additionally add a fraction $u\in[0,\infty)$ of the boundary $A$-reserve
(i.e., $u=U^+$). Then the total value of the LP's AMM position, measured at price $P$ and including paid-out fees from the winner and overrun arbitrageurs,
jumps by the multiplicative factor $1+J_{\mathrm{ext}}^{+}(m,u;\gamma)$ where
\begin{equation}
\label{eq:Jext_plus}
J_{\mathrm{ext}}^{+}(m,u;\gamma)
=
\frac{
e^{-\gamma/2}\sqrt{m}\Big((1+u)+\frac{1}{1+u}\Big)\;-\;m\;-\;e^{-\gamma}
}{m+1} 
=
\frac{e^{-\gamma / 2} \sqrt{m} \cdot \frac{u^2}{1+u}-\left(e^{-\gamma / 2}-\sqrt{m}\right)^2}{m+1}
\end{equation}
\end{enumerate}
Moreover, $J_{\mathrm{ext}}^{-}(m,0;\gamma)=J^{-}(m;\gamma)$ and $J_{\mathrm{ext}}^{+}(m,0;\gamma)=J^{+}(m;\gamma)$ recover
Lemma~\ref{lem:Jpm}.
Differentiating with respect to $u$ yields
\[
\partial_u J_{\mathrm{ext}}^{-}(m,u;\gamma)
=
\frac{e^{\gamma/2}\sqrt{m}}{m+1}\cdot\frac{u(2-u)}{(1-u)^2}>0,
\qquad
\partial_{uu}J_{\mathrm{ext}}^{-}(m,u;\gamma)
=
\frac{2e^{\gamma/2}\sqrt{m}}{m+1}\cdot\frac{1}{(1-u)^3}>0,
\]
\end{lemma}
From the above lemma, we can see that $-\left(e^{\gamma / 2}-\sqrt{m}\right)^2$ is the baseline impermanent loss from informed correction.
$e^{\gamma / 2} \sqrt{m} \cdot \frac{u^2}{1-u}$ is the extra benefit created by arbitrageur overrun which grows convexly in $u$, blowing up as $u \uparrow 1$, for full proof see Appendix~\ref{app:proof_for_Jext}

\subsection{LP Wealth Dynamics and Optimal Liquidity Provision in the Extended Model}
\label{subsec:wealth_extended}
Follow the same formualtion as Section~\ref{sec:baseline}. The LP wealth in num\'eraire units
\begin{equation}
\label{eq:wealth_def_ext_final}
W_t := A^{out}_t P_t + R_{A,t}P_t + R_{B,t} + \dot F_t^{N} - D_t,
\end{equation}
where $F_t^N$ is the cumulative \emph{real-noise} fee income satisfying Lemma~\ref{lem:noise_trader_fee}.

\begin{proposition}[Extended LP wealth dynamics]
\label{prop:wealth_extended_final}
Assume (i) $(P_t,v_t)$ follow \ref{eq:sv_price}--\ref{eq:sv_var},
(ii) real-noise fees satisfy Lemma~\ref{lem:noise_trader_fee}, and
(iii) correction blocks arrive via $(N_t^-,N_t^+)$ with intensities $(\lambda_t^-(g),\lambda_t^+(g))$ and have overrun ratios
$(U^-,U^+)$.
Then the LP wealth admits the jump--diffusion decomposition
\begin{equation}
\label{eq:wealth_dyn_ext_final}
dW_t
=
(A^{out}_t+R_{A,t-})\,dP_t
-r_t^{*}D_t\,dt
+\dot F_t^{N}\,dt
+\Delta W_t^-\,dN_t^- +\Delta W_t^+\,dN_t^+,
\end{equation}
where the race-overrun correction jump sizes can be written equivalently as:
$$
\Delta W_t^{\pm} = V_{t-}^{\mathrm{AMM}}(P_t)\,J_{\mathrm{ext}}^{\pm}(M^{\pm},U^{\pm};\gamma),
$$
with $J_{\mathrm{ext}}^{\pm}$ given by Lemma~\ref{lem:Jext}. Equivalently, using compensated processes $d\widetilde N_t^\pm:=dN_t^\pm-\lambda_t^\pm dt$ and Equation~\ref{eq:sv_price},
\begin{align}
\label{eq:wealth_dyn_ext_comp_final}
dW_t
&=
(A^{out}_t+R_{A,t-})\mu P_t\,dt
+(A^{out}_t+R_{A,t-})\sqrt{v_t}\,P_t\,dB_t
-r_t^{*}D_t\,dt
+\dot F_t^{N}\,dt \nonumber\\
&\quad
+\big(\lambda_t^-\Delta W_t^-+\lambda_t^+\Delta W_t^+\big)\,dt
+\Delta W_t^-\,d\widetilde N_t^-+\Delta W_t^+\,d\widetilde N_t^+.
\end{align}
\end{proposition}

With LP's wealth process, we can now characterize the LP's dynamic allocation with CRRA utility, keeping the baseline infinite-horizon objective:
\[
u(w)=\frac{w^{1-\eta}}{1-\eta},\qquad \eta>0,\ \eta\neq 1,\qquad \rho>0.
\]

\paragraph{Reduced-form control.}
To parallel the baseline CRRA analysis, we represent the LP's control as a scalar $\theta_t\in[\theta_{\min},\theta_{\max}]$ which
\emph{scales the pool depth} through $K_t=K(\theta_t)$ and therefore changes:
(i) the persistent noise-fee rate via Lemma~\ref{lem:noise_trader_fee},
and (ii) the arbitrageur overrun ratios $U^\pm$ through Proposition~\ref{lem:Jext}.

To keep the CRRA homogeneity structure, we summarize the environment at time $t$ by the conditional marks
$(M^-,U^-)$ and $(M^+,U^+)$ and intensities $(\lambda_t^-,\lambda_t^+)$, and we define proportional noise-fee yield on AMM inventory:
\begin{equation}
\label{eq:psi_noise_ext_final}
\psi_t^{N}
:=
\frac{\mathbb E[\dot F_t^{N}\mid \mathcal F_{t-}]}{V_{t-}^{\mathrm{AMM}}(P_t)},
\qquad
\text{and it scales like }\psi^N(\theta,\cdot)\ \propto\ K^{\beta-\frac12}\ \propto\ \theta^{2\beta-1}.
\end{equation}

\begin{theorem}[CRRA HJB reduction and \emph{correct} first/second derivatives when $U$ depends on $\theta$ via $K(\theta$)]
\label{thm:crra_extended_endogK}
The LP solves
\begin{equation}
\label{eq:crra_problem_ext_final}
V(w)=\sup_{\{\theta_t\}}\ \mathbb{E}\!\left[\int_0^\infty e^{-\rho t}\frac{W_t^{1-\eta}}{1-\eta}\,dt\ \Big|\ W_0=w\right],
\qquad
\theta_t\in[\theta_{\min},\theta_{\max}],
\end{equation}
subject to the wealth dynamics in Proposition~\ref{prop:wealth_extended_final}.

Assume (as in the baseline CRRA reduction) that the controlled wealth can be written in proportional form
\begin{equation}
\label{eq:wealth_prop_form_ext_final}
\frac{dW_t}{W_{t-}}
=
\Big(\mu-r^*\theta_t+\theta_t\psi_t^N\Big)dt
+\sqrt{v_t}\,dB_t
+\theta_t J_{\mathrm{ext}}^{-}(M^-,U^-(\theta_t, v_t);\gamma)\,dN_t^-
+\theta_t J_{\mathrm{ext}}^{+}(M^+,U^+(\theta_t, v_t);\gamma)\,dN_t^+,
\end{equation}
where $U^\pm$ captures that arbitrageur overrun scales with pool depth and outside volatility because $\theta$ determines $K$. Define the admissible set $\Theta_{\mathrm{adm}}$ as in Theorem~\ref{thm:crra_hjb_cfmm}.

Then, conditional on $\mathcal F_{t-}$, the CRRA HJB maximization reduces pointwise in $\theta$ to maximizing:
\begin{align}
\label{eq:Phi_ext_endogK}
\Phi_t(\theta)
&:=
\mu-r^*\theta-\frac{\eta}{2}v_t + \theta\psi_t^N \nonumber\\
&\quad
+\sum_{\pm}\frac{\lambda_t^{\pm}(\theta, v)}{1-\eta}\,
\mathbb{E}\!\left[(1+\theta J_{\mathrm{ext}}^{\pm}(M^{\pm},U^{\pm}(\theta, v);\gamma))^{1-\eta}-1\ \big|\ \mathcal F_{t-}\right] 
\end{align}

Then we have that the existence, uniquenes, and interior FOC results similar to the baseline model in Theorem~\ref{thm:crra_hjb_cfmm} with more complicated expressions and parameter restrictions which we will leave to the full theorem statement and proof in Appendix~\ref{app:full_thm_and_proof_extended}, denote the interior optimal liquidity provision by $\theta^*$
\end{theorem}

\subsection{Hump-Shaped Optimal Liquidity Provision $\theta(v)$: Sufficient Conditions}
\label{subsec:hump_theorem}

Let $F(\theta ; \zeta) = \Phi'(\theta) =0$, where $\zeta$ is some parameter or state variable (e.g., $r^*, v_t, \sigma_N, \beta, \gamma, \ldots$ ). In the region where the objective is locally concave in $\theta$, such that $F_\theta\left(\theta^*\right) < 0$, then we have 
$$
\frac{d \theta^*}{d \zeta}=-\frac{F_\zeta\left(\theta^*\right)}{F_\theta\left(\theta^*\right)} .
$$
then $F_\theta=\Phi^{\prime \prime}<0$, and the sign is simply:
- $\operatorname{sign}\left(d \theta^* / d \zeta\right)=\operatorname{sign}\left(F_\zeta\right)$.

Throughout this subsection we suppress inessential constants and focus on: 
(i) the persistent noise-fee term,
(ii) the correction-jump term where winner IL offset by overrun order flow, and (iii) the pool-size elasticity $\beta$.
and we want to pin down how the optimal $\theta^*$ depends on these three terms, assuming the representative LP supplies the entire pool, we have 
$K(\theta)=\bar K\,\theta^2,\qquad \bar K>0$.
With gas as a volatility driven function $g_t=g(v_t)$ that is continuous, increasing, and convex.

\begin{proposition}[Overrun fraction inherits $\sqrt{v}$ and $K^{\beta-\frac12}$ scaling]
\label{prop:u_scaling_cs}
Fix $(P_t,v_t)$ at time $t$. the conditional expected buy-side overrun fraction satisfies Lemma~\ref{lem:Jext}
Under $K(\theta)=\bar K\theta^2$, the expected value becomes
\begin{equation}
\label{eq:u_theta_power}
\mathbb E[u_t^-\mid \mathcal F_{t-}]
=
\widetilde C^-_t\;\theta^{2\beta-1},
\qquad
\widetilde C^-_t:=C^-\bar K^{\beta-\frac12}\sqrt{v_t}\sqrt{P_t}.
\end{equation}
Hence $\partial_\theta \mathbb E[u_t^-]$ has the sign of $(2\beta-1)$.
Analogous statements hold on the sell-side with a constant $C^+=\sqrt{2/\pi}\,e^{+\gamma/2}$.
\end{proposition}

We now formalize the hump-shape logic using the implicit function theorem on the LP's interior FOC.
Suppressing terms not central to this subsection, define the reduced first-order condition at time $t$:
\begin{equation}
\begin{aligned}
\label{eq:FOC_reduced_cs}
F(\theta,v)
:=
\Big(\psi^N(\theta,v)+\theta\,\partial_\theta\psi^N(\theta,v)\Big)
\;+\;
\sum_{\pm}\lambda^\pm(\theta,v)\,
\mathbb E\!\Big[H^\pm(\theta,v)\mid\mathcal F_{t-}\Big] \\
+ \frac{\partial \lambda^{ \pm}}{\partial \theta} \frac{1}{1-\eta} \mathbb{E}\left[\left(1+\theta J^{ \pm}\right)^{1-\eta}-1\right]
\;-\;r^* 
\;=\;0,
\end{aligned}
\end{equation}
where $H^\pm(\theta,v):=(J_{\mathrm{ext}}^\pm+\theta\partial_\theta J_{\mathrm{ext}}^\pm)\,(1+\theta J_{\mathrm{ext}}^\pm)^{-\eta}$,
and where $\psi^N$ is defined by $\psi^N(\theta,v):=\frac{\mathbb E[\dot F^N\mid \mathcal F_{t-}]}{V^{\mathrm{AMM}}_{t-}(P_t)}$,
We assume gas affects both (i) noise participation (Lemma~\ref{lem:noise_trader_fee}) and
(ii) profitable arbitrage/correction frequency via a profitability threshold, so that
\begin{equation}
\label{eq:lambda_gas_monotone}
\lambda^\pm(\theta,v)=\bar\lambda^\pm\cdot \mathbb P\!\left(|\Delta|\ge \sqrt{\frac{2g(v)}{K(\theta)^\beta}}\ \Big|\ v\right)
\end{equation}
which is decreasing in $v$ whenever $g(v)/v$ is increasing (``gas grows faster than volatility'').

\begin{theorem}[Hump-shaped $\theta^*(v)$ under convex gas and $\beta>\tfrac12$]
\label{thm:hump_theta_v}
Fix time $t$ and treat $(P_t,Q_t)$ as given state variables. Assume:
\begin{enumerate}
\item \textbf{(Concavity / stability)} There is an interior optimum $\theta^*(v)\in(\theta_{\min},\theta_{\max})$ solving $F(\theta,v)=0$ and
$F_\theta(\theta^*(v),v)<0$, the explicit condition for this to be true can be found in Appendix~\ref{app:full_thm_and_proof_extended}.
\item \textbf{(Gas grows faster than volatility)} $g$ is increasing, convex, and $g(v)/v$ is increasing with $g(v)\to\infty$ as $v\to\infty$.
\item \textbf{(Overrun elasticity)} $\beta>\tfrac12$, so the overrun fraction $u(\theta,v)$ increases in both $\theta$ and $\sqrt v$~ (Proposition~\ref{prop:u_scaling_cs}).
\item \textbf{(Low-vol dominance of the severity channel)} There exists $v_0>0$ such that $F_v(\theta^*(v_0),v_0)>0$.
\end{enumerate}
Then:
\begin{enumerate}
\item (\emph{Rising region}) In a neighborhood of $v_0$, $\theta^*(v)$ is increasing in $v$:
\[
\frac{d\theta^*}{dv}(v_0)= -\frac{F_v}{F_\theta}\Big|_{(\theta^*(v_0),v_0)} \;>\;0.
\]
\item (\emph{Falling tail}) For sufficiently large $v$, $\theta^*(v)$ is decreasing in $v$, and $\theta^*(v)\to\theta_{\min}$ as $v\to\infty$.
\item (\emph{Hump shape}) Consequently, there exists at least one $v^*\in(v_0,\infty)$ at which $\theta^*(v)$ attains a maximum
(i.e.\ $\theta^*(v)$ is hump-shaped).
\end{enumerate}
\end{theorem}
At low volatility, price movements are small and gas costs are low, so noise traders remain active and generate steady fee income. As volatility rises, arbitrage overruns become larger and more valuable, and—when $\beta>\tfrac12$—this overrun channel initially dominates. 
Even though impermanent loss increases, the LP finds it optimal to provide more liquidity to capture the higher overrun arbitrage flow. 
At high volatility, this logic reverses. Convex gas costs sharply reduce trading activity, while price jumps become larger and less frequent, creating substantial impermanent loss when they occur. As a result, the marginal return to supplying liquidity falls, and the LP optimally reduces liquidity provision. 
Together, these forces imply that optimal liquidity provision yields a hump-shaped optimal solution $\theta^*(v)$. 
See Appendix~\ref{app:proof_for_hump} for full proof. 
We also validated this equilibrium liquidity provision behavior in ETH/USDC pool where liquidity providers collectively show a hump-shaped liquidity provision behavior in Figure~\ref{fig:hump_shape}.

\section{Conclusion}
 
This paper studies liquidity provision in CFMM-based AMMs through a unified theoretical and empirical framework that highlights the interaction between informed arbitrage, trading frictions, and endogenous liquidity supply. 
In the baseline model, we show that informed arbitrageurs face an intrinsic buy–sell asymmetry when trading against a CFMM-AMM. This asymmetry implies that, in a market populated exclusively by informed arbitrageurs, liquidity provision is strictly dominated for LPs, even in the presence of trading fees. 

Guided by these theoretical insights, we then turn to transaction-level data to empirically examine the trading behavior observed in real CFMM-AMM markets. The data strongly corroborate the model’s prediction of buy–sell asymmetry, revealing persistent differences in volume, frequency, and profitability across trade directions. 
Moreover, by classifying trades into profitable and unprofitable transactions, we uncover substantial heterogeneity in trading behavior that closely mirrors patterns documented in traditional financial markets: profitable trades tend to be larger but less frequent, while unprofitable trades occur more frequently and contribute disproportionately to overall transaction counts. 
Importantly, we also show that gas fees, even though depends on the congestion of the entire Ethereum network, are not all exogenous, but instead are partially linked to CEX price volatility, providing empirical support for modeling gas costs as an endogenous, volatility-dependent friction.

To reconcile these empirical patterns with equilibrium liquidity provision, we develop an extended model that incorporates time-varying price volatility, volatility-dependent gas fees, the presence of genuine noise traders, and a class of overrun arbitrageurs whose trades are ex-post unprofitable despite being ex ante profit-maximizing. 
Within this richer environment. LPs optimally respond to these forces by balancing fee revenue against impermanent loss and adverse selection, leading to an endogenous liquidity provision decision that varies non-monotonically with volatility. In equilibrium, liquidity provision exhibits a hump-shaped relationship with price volatility: moderate volatility attracts liquidity by increasing fee-generating opportunities, while high volatility induces liquidity withdrawal as risk and arbitrage pressure dominate.

Taken together, our results highlight that AMM liquidity provision cannot be understood in isolation from arbitrage competition, trading frictions, and volatility dynamics. Rather than viewing unprofitable trades as purely exogenous noise, our analysis shows that they are an endogenous outcome of competitive arbitrage under realistic frictions, and that they play a central role in sustaining liquidity provision in practice. In the future, this framework can be extended to alternative AMM designs--including concentrated liquidity or dynamic fee mechanisms.

\newpage
\bibliographystyle{ACM-Reference-Format}
\bibliography{references}

\newpage
\appendix

\section{Supplementary Materials}

\subsection{Supplementary Figures}
\FloatBarrier

\begin{figure}[htbp]
    \centering
    \includegraphics[width=0.85\linewidth]{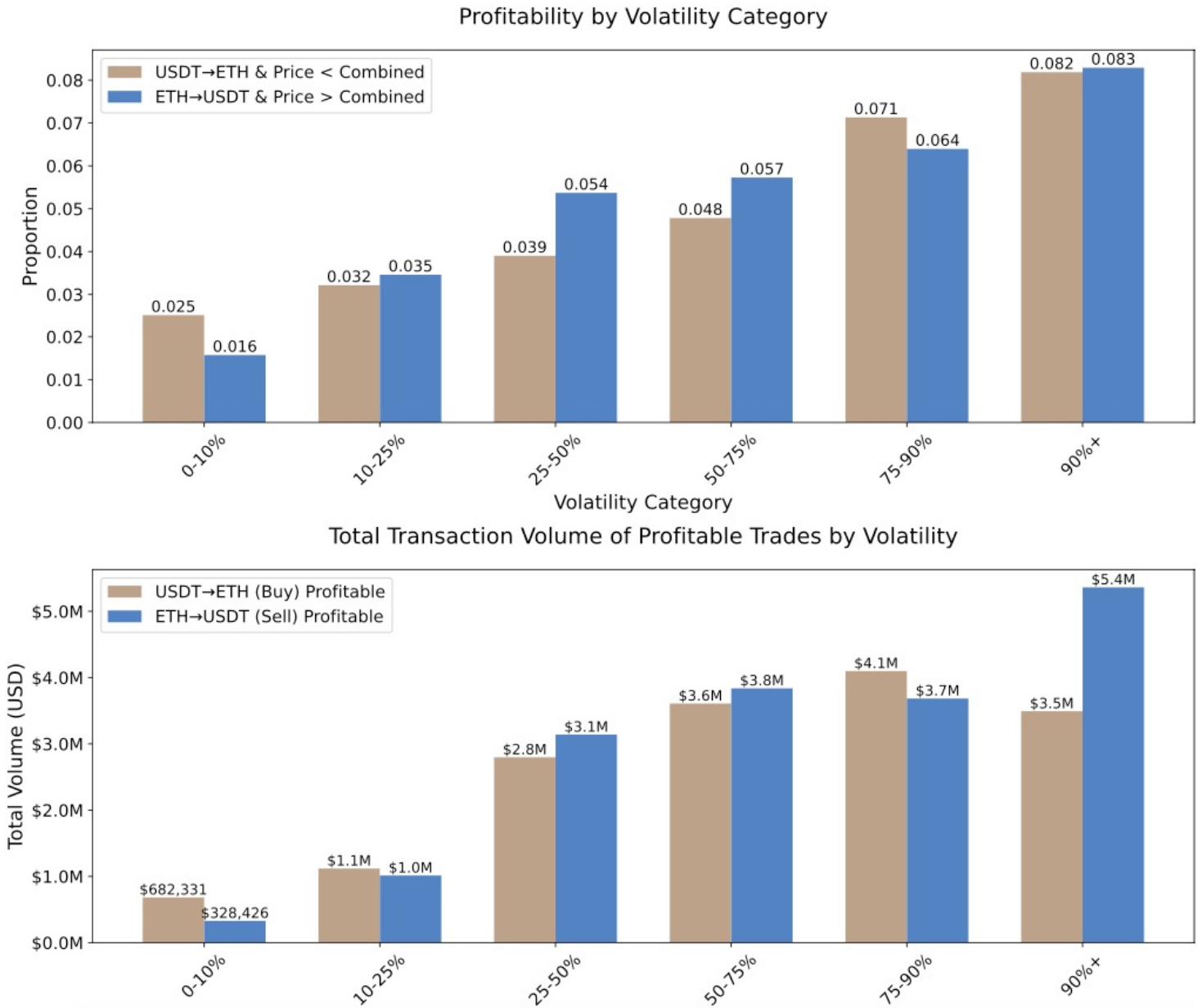}
    \caption{Buy-Sell Asymmetry in Profitability Rate of ETH/USDT pool by Volatility}
    \label{fig:ETH_asymmetry}
\end{figure}

\begin{figure}[htbp]
    \centering
    \includegraphics[width=.8\linewidth]{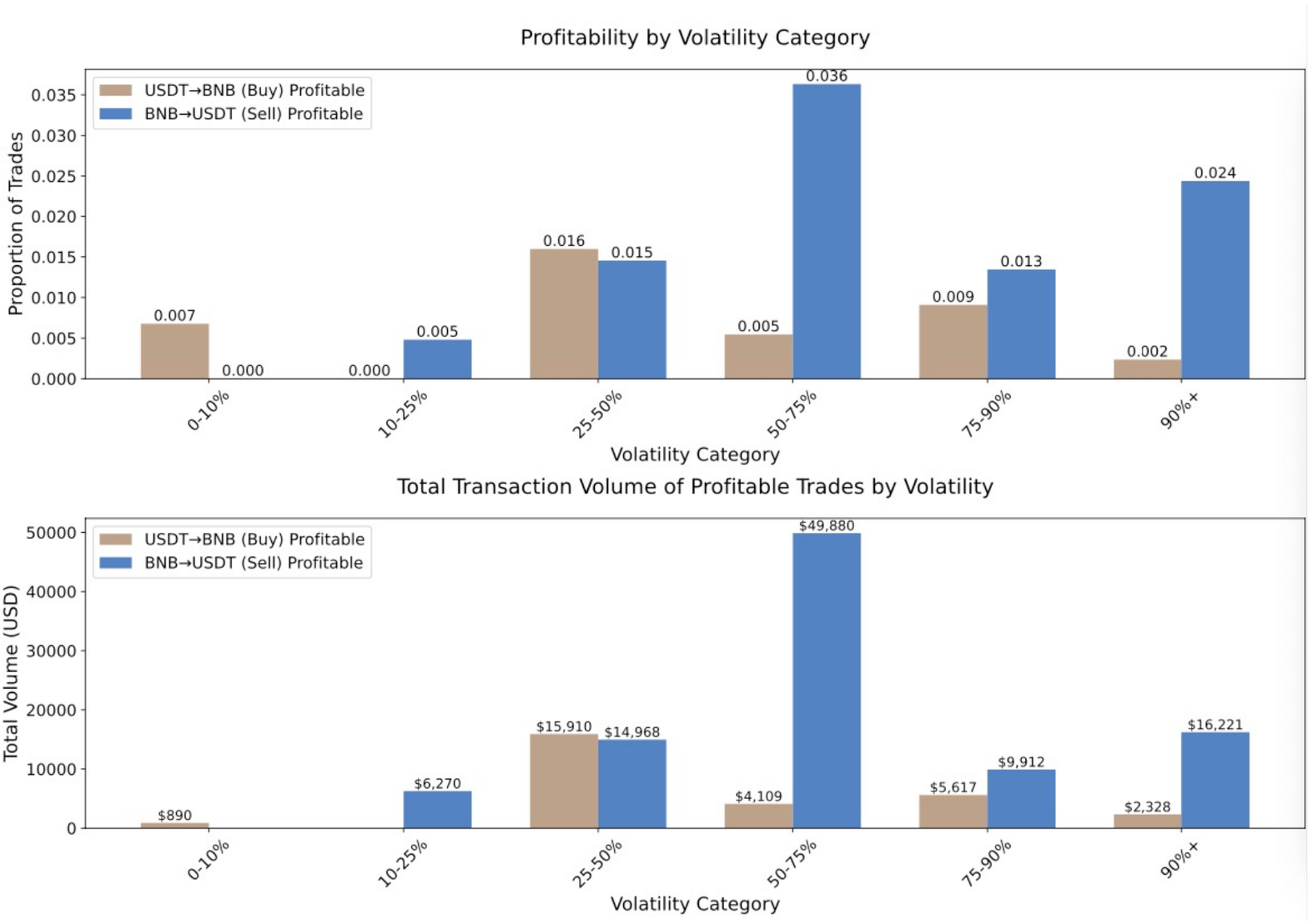}
    \caption{Buy-Sell Asymmetry in Profitable Transactions in BNB/USDT Pool}
    \label{fig:BNB_asymmetry}
\end{figure}

\begin{figure}[htbp]
    \centering
    \includegraphics[width=1\linewidth]{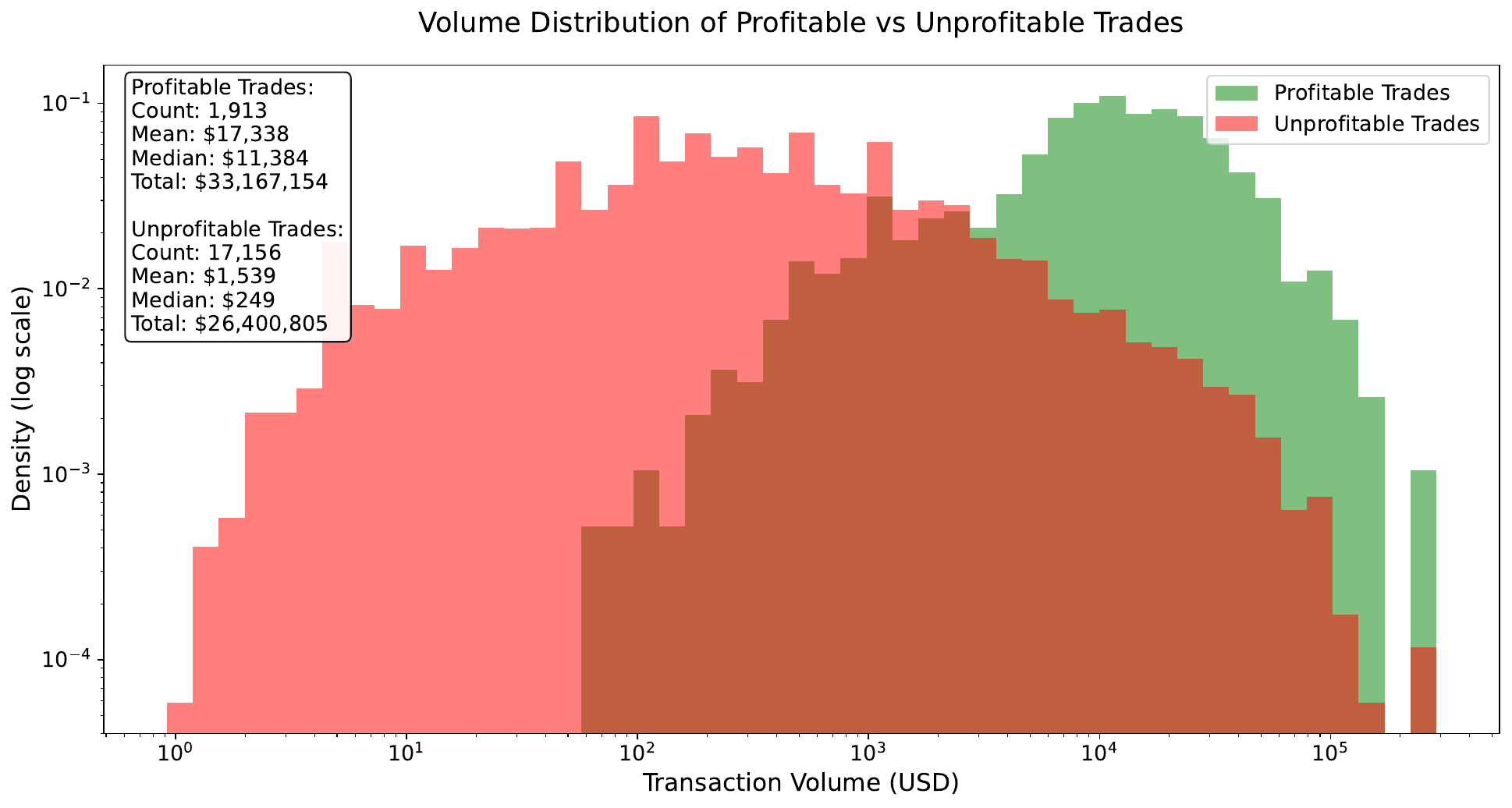}
    \caption{Volume Distribution of Profitable (Green) v.s. Unprofitable (Red) Transactions ETH/USDT pool}
    \label{fig:ETH_Noise}
\end{figure}

\begin{figure}[htbp]
    \centering
    \includegraphics[width=1\linewidth]{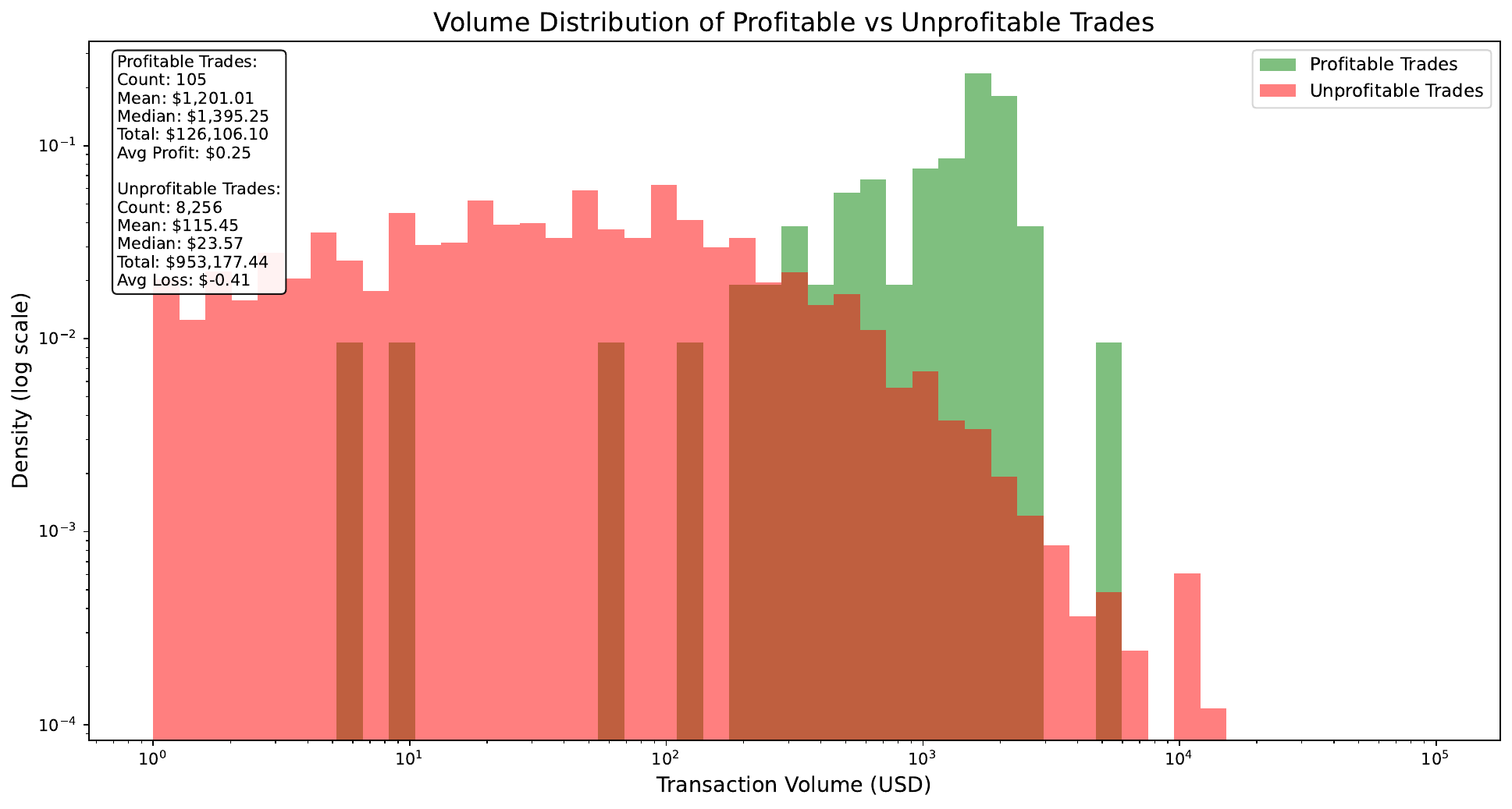}
    \caption{Volume Distribution of Profitable (Green) v.s. Unprofitable (Red) Transactions BNB/USDT pool}
    \label{fig:BNB_Noise}
\end{figure}

\begin{figure}[htbp]
    \centering
    \includegraphics[width=1\linewidth]{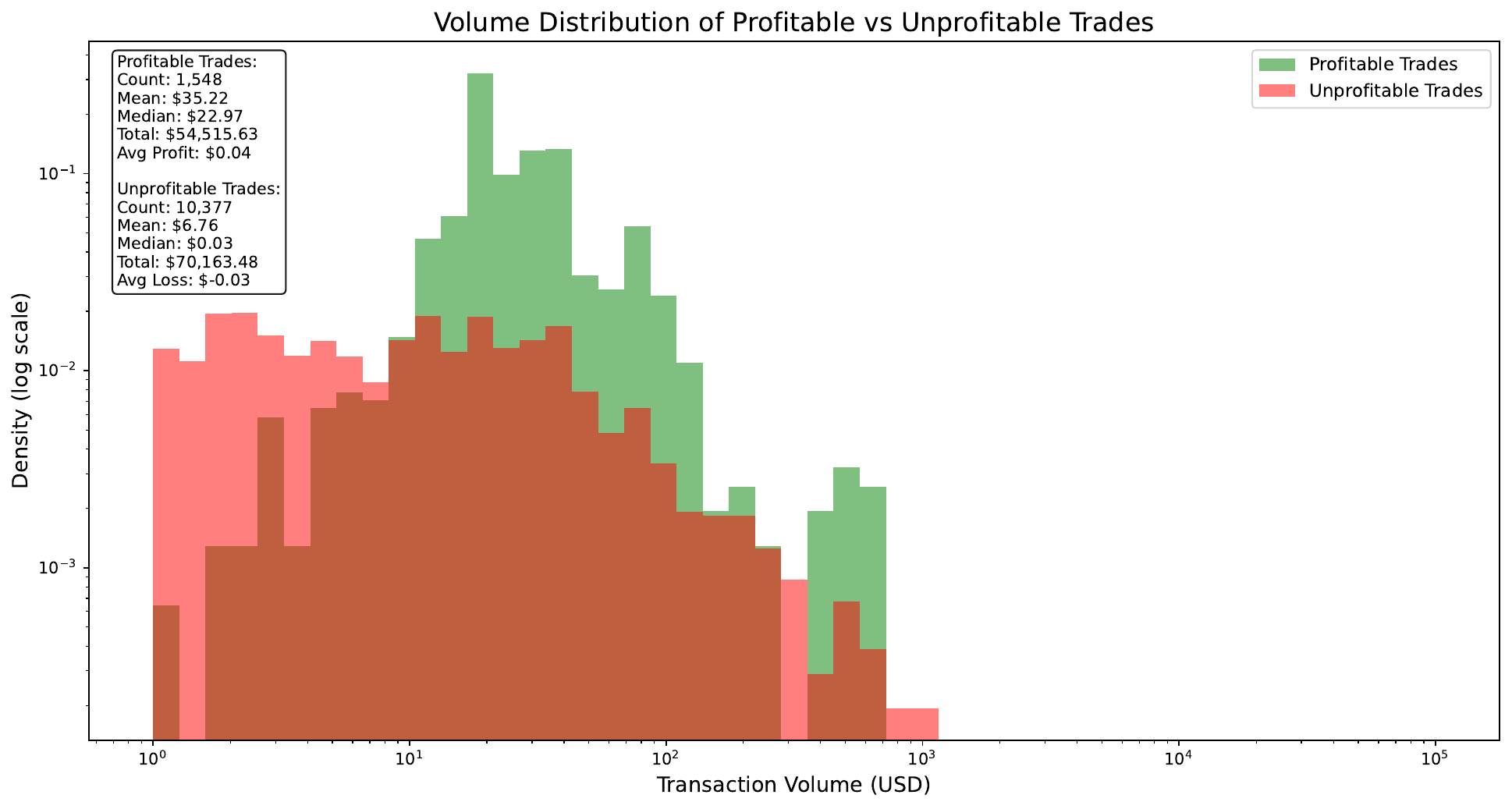}
    \caption{Volume Distribution of Profitable (Green) v.s. Unprofitable (Red) Transactions POL/USDT pool}
    \label{fig:POL_Noise}
\end{figure}

\begin{figure}[htbp]
    \centering
    \includegraphics[width=0.8\linewidth]{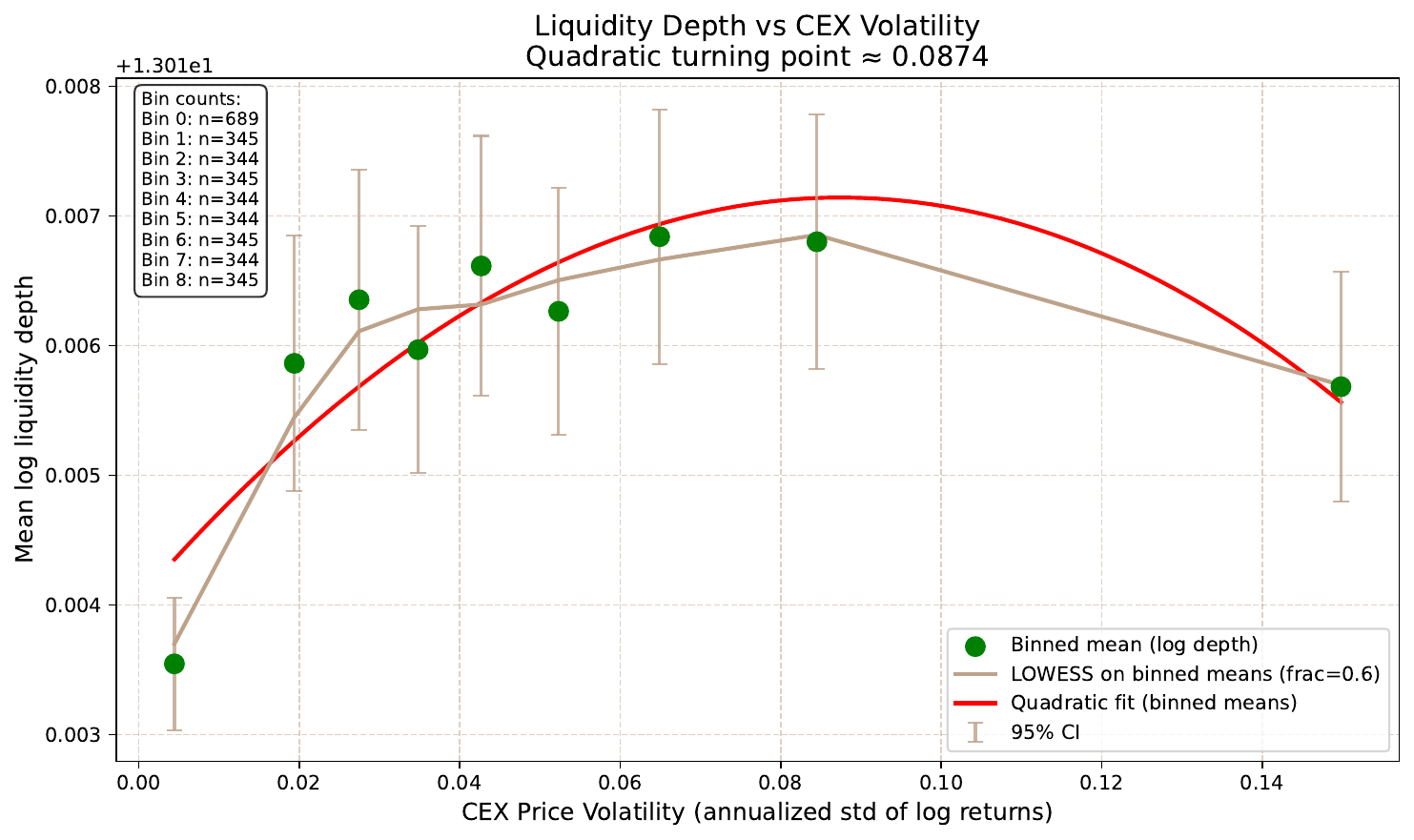}
    \caption{Hump-shape Liquidity Provision in USDC/ETH}
    \label{fig:hump_shape}
\end{figure}

\FloatBarrier


\subsection{Supplementary Tables}
\FloatBarrier

\begin{table}[htbp]
\centering
\caption{Analysis of Gas Price Determinants: Granger Causality and Regression Results}
\label{tab:regression_and_causality}

\begin{tabular}{l}
\textbf{Panel A: Granger Causality Tests} \\
\end{tabular}

\vspace{0.2cm}

\begin{tabular}{c|cc|cc}
\hline
& \multicolumn{2}{c|}{Price Gap → Gas} & \multicolumn{2}{c}{Volatility → Gas} \\
Lag & F-stat & p-value & F-stat & p-value \\
\hline
1 & 0.2583 & 0.6113 & 565.9885 & $<0.0001^{***}$ \\
2 & 0.2316 & 0.7933 & 131.1265 & $<0.0001^{***}$ \\
3 & 0.3485 & 0.7902 & 50.8777 & $<0.0001^{***}$ \\
4 & 0.2679 & 0.8987 & 27.0696 & $<0.0001^{***}$ \\
5 & 0.2888 & 0.9194 & 16.4579 & $<0.0001^{***}$ \\
6 & 0.2891 & 0.9424 & 11.1980 & $<0.0001^{***}$ \\
\hline
\end{tabular}

\vspace{0.5cm}

\begin{tabular}{l}
\textbf{Panel B: Regression Analysis of Volatility Effects} \\
\end{tabular}

\vspace{0.2cm}

\begin{tabular}{l|c|c|c}
\hline
Volatility Window & 100 & 150 & 200 \\
\hline
Coefficient & 857.41*** & 801.45*** & 771.82*** \\
            & (18.30) & (16.10) & (14.60) \\
R² & 0.1032 & 0.1152 & 0.1284 \\
Adjusted R² & 0.1032 & 0.1152 & 0.1284 \\
P-value & $<0.0001^{***}$ &$<0.0001^{***}$ &$<0.0001^{***}$ \\
\hline
\end{tabular}

\vspace{0.3cm}

\begin{minipage}{0.95\linewidth}
\footnotesize
\textit{Note:} Panel A shows Granger causality tests comparing predictive power of price ratios versus volatility.
Panel B presents regression results of gas prices on volatility calculated over different time windows.
Coefficients in Panel B are in billions ($\times 10^9$) wei with standard errors in parentheses.
*** indicates significance at 0.1\% level.
\end{minipage}
\end{table}

\begin{table}[htbp]
\centering
\caption{Granger Causality Tests: Volatility $\rightarrow$ Unprofitable Trading Volume (p-values)}
\label{tab:unprofitable_vol_comove}
\begin{tabular}{lccc}
\hline
Lag 
& ETH/USDT (30 mins) 
& POL/USDT (3 mins) 
& BNB/USDT (3 mins) \\
\hline
1 & $<0.0001^{***}$ & $0.0060^{***}$ & $0.4853$ \\
2 & $<0.0001^{***}$ & $0.0112^{**}$  & $0.0810^{*}$ \\
3 & $<0.0001^{***}$ & $0.0191^{**}$  & $0.0180^{**}$ \\
4 & $<0.0001^{***}$ & $0.0428^{**}$  & $0.0116^{**}$ \\
5 & $<0.0001^{***}$ & $0.0379^{**}$  & $0.0216^{**}$ \\
\hline
\multicolumn{4}{l}{\footnotesize Significance: $^{***}p<0.01$, $^{**}p<0.05$, $^{*}p<0.10$.}
\end{tabular}

\vspace{0.3cm}

\begin{minipage}{0.95\linewidth}
\footnotesize
\textit{Note:} This table shows Granger causality of time-aggregated price volatility on unprofitable transaction volume across 3 different pool. ETH/USDT pool transactions are aggregated at 30mins level; POL/USDT and BNB/USDT transactions are aggregated at 3mins level.
\end{minipage}
\end{table}

\FloatBarrier

\clearpage
\subsection{Remarks, Proofs, and Supplementary Theorems}
\FloatBarrier
\subsubsection{Remark: Compounding Fee v.s. Paid-out Fee}
\label{remark:why_paid_out}
In the baseline analysis of \S4, we adopt a "fee-extracted" formulation in which transaction fees are paid out to LPs and do not enter the pool reserves. Under this normalization, the constant-product invariant $K$ is unchanged by swaps, which allows us to isolate the effect of informed arbitrage on LP wealth and derive clean expressions for jump returns in Lemma~\ref{lem:Jpm} and the dominance result in Theorem~\ref{thm:crra_hjb_cfmm}.

In Uniswap V2 and similar protocols, transaction fees are instead retained within the pool, mechanically increasing reserves and causing the invariant $K$ to grow over time. Importantly, this institutional difference does not alter the economic content of our results. 
From the LP's perspective, fee compounding and fee payout are equivalent up to a timing and accounting transformation: total LP wealth after a trade can always be decomposed into (i) the mark-to-market value of reserves at the external price and (ii) cumulative fee revenue, regardless of whether fees are held inside or outside the pool.

In particular, under purely informed trading, winning arbitrageurs generate adverse selection losses for LPs through price impact. 
Fee compounding partially offsets these losses by increasing reserves, but does not reverse them. 
If one extracts the accumulated fees from the pool immediately after each trade-or equivalently tracks fees in a separate account-the resulting wealth dynamics coincide with the fee extracted model analyzed here. Consequently, the strictly negative jump returns in Lemma~\ref{lem:Jpm} and the no-liquidity benchmark in Theorem~\ref{thm:crra_hjb_cfmm} describe the same economic mechanism that operates in fee-compounding AMMs, abstracting from mechanical reinvestment of fees.

\subsubsection{Proof of Proposition \ref{prop:optimal_swap}}
\label{app:optimal_swap_baseline}
\begin{proof} 
Firstly, we will calculate a general expression of $\Delta B$ in terms of $\Delta A$. 
By the constant product market maker rule, we have 
$$
\left(R_A+\Delta A\right)\left(R_B+\Delta B\right)-R_A R_B=0
$$
Expanding and simplifying the expression, we obtain 
$$
\begin{gathered}
R_A R_B+R_A \Delta B+\Delta A R_B+\Delta A \Delta B-R_A R_B=0 \\ 
R_A \Delta B+\Delta A R_B+\Delta A \Delta B=0 \\
\Delta B=-\frac{\Delta A R_B}{R_A+\Delta A}
\end{gathered}
$$
\textbf{Case 2.} when $P_t<Q(t^{-}) e^{-\gamma}$, more specifically, $P_t < \frac{R^A_{t-}}{R^B_{t-}} e^{-\gamma}$, the profit function $\pi$ is, 
\begin{equation}
\begin{aligned}
\pi(\Delta A) &= -\Delta A \cdot P_t-e^{-\gamma} \Delta B \\
&= -\Delta A \cdot P_t + e^{-\gamma}(\frac{\Delta A R_B}{R_A+\Delta A})
\end{aligned}
\end{equation}
Then we take the first order derivative of the profit function w.r.t. $\Delta A$, giving that
$$
\begin{aligned}
\frac{d \pi(\Delta A)}{d \Delta A} &= -P_t + e^{-\gamma} \left(\frac{ R^B_{t-}(R^A_{t-} + \Delta A) + \Delta A R^B_{t-}}{(R^A_{t-}+\Delta A)^2}\right) \\
&= -P_t + \frac{e^{-\gamma}R^B_{t-}R^A_{t-}}{(R^A_{t-} + \Delta A)^2} 
\end{aligned}
$$
Setting the RHS of the equation to 0, we obtain
$$
\begin{aligned}
P_t = \frac{e^{-\gamma}R^B_{t-}R^A_{t-}}{(R^A_{t-} + \Delta A)^2} \\
R^A_{t-} + \Delta A = \sqrt{\frac{e^{-\gamma}R^B_{t-}R^A_{t-}}{P_t}} \\
\Delta A = \sqrt{\frac{e^{-\gamma}k}{P_t}} - R^A_{t-} 
\end{aligned}
$$
Therefore, we obtain the expression for optimal $\Delta A^* $ when selling. To examine the validity of the first order condition, we calculate the second order derivative of the profit function(we omit $t$ for simplicity here),  which gives 
$$
\frac{d\pi^2(\Delta A)}{d(\Delta A)^2}=-2 R_B e^{-\gamma} \frac{\Delta A}{\left(R_A+\Delta A\right)^3}
$$
Since under any condition, $(R_A + \Delta A) > 0$, we have that when the arbitrageur buys asset $A$ from the AMM($\Delta A <0$), the profit function is concave up, thus $\Delta A^{*}$ is only an optimal solution when the arbitrageur is selling asset $A$ to the AMM, thus we write 
$$
\Delta A^{*}_{s} = \sqrt{\frac{e^{-\gamma}k}{P_t}} - R^A_{t-}
$$
And since the arbitrageur cannot sell asset $A$ more than what she owns, we have that 
$$
\Delta A_{s}^*= \min\{\sqrt{\frac{k e^{-\gamma}}{P_t}}-R^A_{t-}, A_t\}
$$ 
\textbf{Case 1.} When the arbitrageur buys from the AMM, his optimal decision is determined by three conditions: 1. \textbf{Price Matching Condition}, 2. \textbf{Availability of asset $B$}, and 3. \textbf{AMM Reserves}, \\
Consider the first condition, the marginal price of buying asset A is given by
$$
Q'=\frac{R_B e^{+\gamma}}{\left(R_A+\Delta A\right)}
$$
We set the expression equal to the market price $P$, obtain
$$
\Delta A=\frac{R_B e^{\gamma}}{P}-R_A<0 \Rightarrow \frac{R_B e^{\gamma}}{P} < R_A
$$
For condition 2, from the previous calculation, we have that 
$$
-\frac{R_B \Delta A}{R_A+\Delta A} e^{+\gamma} \leq B
$$
And for condition 3, we have
$$
-\Delta A < R_A
$$
With those three conditions in hand, we can now compute the optimal buying amount for the arbitrageur. First, we compute the maximum $\Delta A$ satisfying the price matching condition:
$$
\Delta A_{\max }=\frac{R_B e^{+\gamma}}{P_t} - R_A
$$
If $\Delta A_{\max } \ge 0$, no profitable trade exists, and if $\Delta A_{\max }<0$, we now turn to condition 2 to compute the maximum $\Delta A$ the arbitrageur can afford $\Delta A_{\text {aff}}$, which is the solution of $\Delta A$ for the equality below 
$$
-\frac{R_B \Delta A}{R_A-\Delta A} e^{+\gamma}=B_t
$$
Solve for $\Delta A_{\text{aff}}$ :
$$
\Delta A_{\text{aff}}=-\frac{B_t\left(R_A\right)}{R_B e^{+\gamma}+B_t}
$$
Thus we can determine the optimal $\Delta A$
$$
\Delta A^*=\max \left\{\Delta A_{\max }, \Delta A_{\text{aff}}\right\}
$$
the reason we put a maximum is because both values are negative, and therefore we conclude our proof.
\end{proof}

\subsubsection{Proof of Proposition \ref{prop:b_s_asymmetry}}
\label{app:proof_for_b_s_asymmetry}
Based on the constant product market-making rule, for consistency, we call the risky asset asset A and the numeraire asset B. The AMM's price of asset A is determined by $R_A/RB$ given that $f(x, y) = x \cdot y = k$. \\
Now suppose that there are symmetric price changes in the CEX such that at $P_{t-} = Q_{t-}$, and at time $t$, we have $P_{t+} - Q_{t-} = Q_{t-} - P_{t-}$, then for the AMM price to converge to a higher price, given that $k$ does not change, we can calculate the rate of convergence per unit change in the pool by calculating the derivative of $R_A$'s AMM price $\frac{R_B}{R_A}$ relative to $R_A$, which is 
$$
\frac{d\frac{R_B}{d R_A}}{dR_A} = -\frac{2k}{R^3_A} = - \frac{2R_B}{R^2_A}
$$
From the simple derivation above, we can see that the absolute change of price for asset $A$ is determined by 
$$
\frac{2R_B}{R_A^2}
$$
Therefore, the more expensive the asset(mathematically the larger the fraction $R_B/R_A$, the larger price impact there is for addition of one unit of risk asset addition to the pool. 

\begin{figure}
    \centering
    \includegraphics[width=0.9\linewidth]{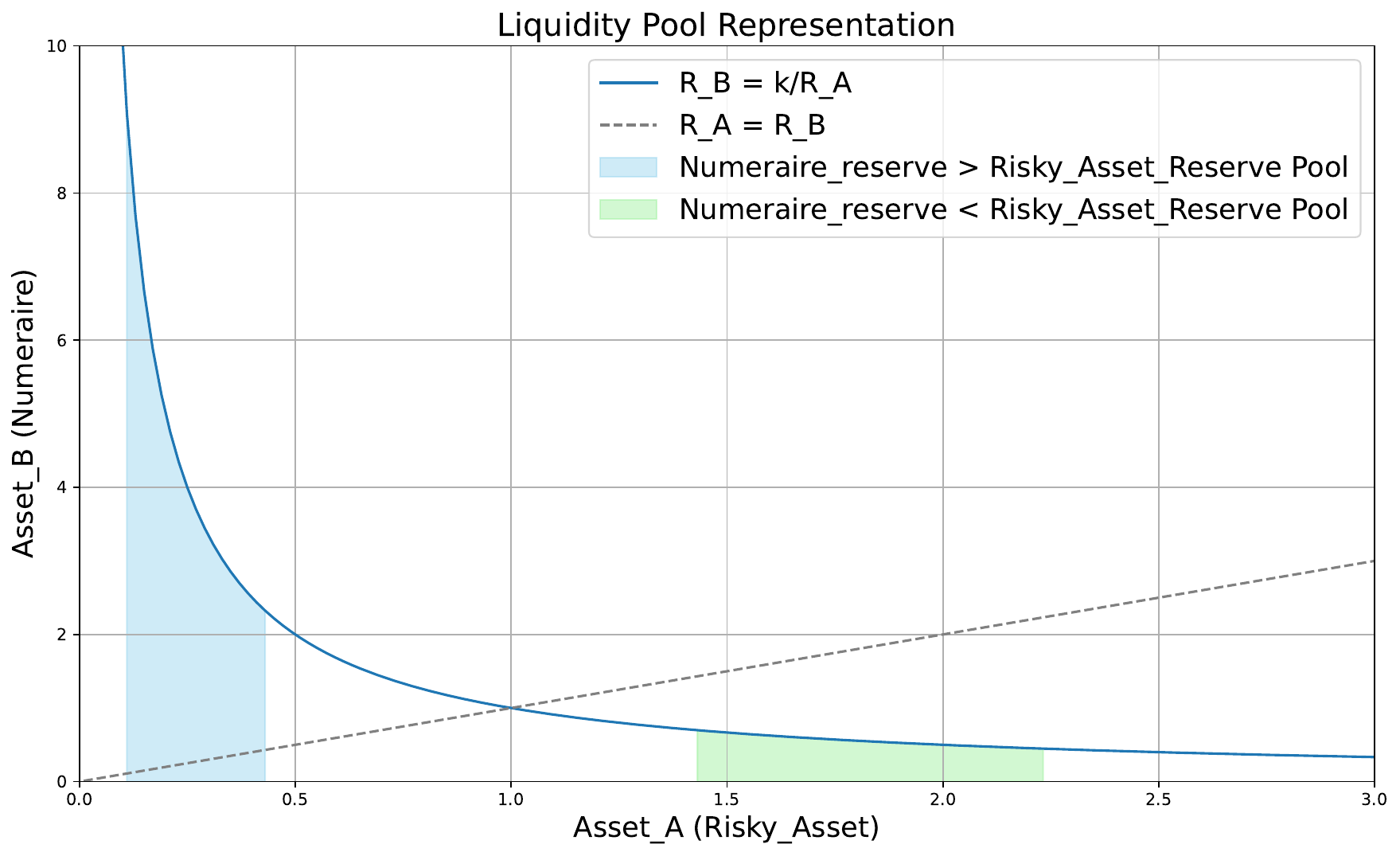}
    \caption{Illustration of Buy/Sell Asymmetry}
    \label{fig:liquidity_compare}
\end{figure} 

\subsubsection{Proof of Lemma~\ref{lem:Jpm}}
\label{app:proof_for_lemma_Jpm}
\begin{proof}
\textbf{Buy-side correction.}
Let $K=R_AR_B=R_A^2Q_-$ and $m=P/Q_-$. A buy-side correction sets $Q'=Pe^{-\gamma}$ while keeping $K$ fixed, hence
\[
R^{A'}=\sqrt{\frac{K}{Q'}}=R_A\sqrt{\frac{Q_-}{Pe^{-\gamma}}}=R_A\sqrt{\frac{e^{\gamma}}{m}},
\qquad
R^{B'}=\sqrt{KQ'}=\sqrt{R_A^2Q_-Pe^{-\gamma}}=R_B\sqrt{me^{-\gamma}}.
\]
The net amount of $B$ added to reserves is $y_{\mathrm{net}}:=R^{B'}-R_B$.
Under a fee-extracted rule, the trader pays total $y_{\mathrm{tot}}=e^{\gamma}y_{\mathrm{net}}$, so the paid-out fee is
$y_{\mathrm{fee}}:=y_{\mathrm{tot}}-y_{\mathrm{net}}=(e^{\gamma}-1)(R^{B'}-R_B)$.
Thus the LP's post-correction total value at price $P$ is
\[
V_{\mathrm{post}}
:=PR^{A'}+ R^{B'} + y_{\mathrm{fee}}
=PR^{A'} + R^{B'} + (e^{\gamma}-1)(R^{B'}-R_B).
\]
Using $R_AP=R_B m$ and $R^{B'}-R_B=R_B(\sqrt{me^{-\gamma}}-1)$, we obtain
\[
PR^{A'} + R^{B'}
= R_A P\frac{e^{\gamma/2}}{\sqrt{m}} + R_A P\frac{e^{-\gamma/2}}{\sqrt{m}}
= R_B\sqrt{m}\big(e^{\gamma/2}+e^{-\gamma/2}\big),
\]
and
\[
y_{\mathrm{fee}}
=(e^{\gamma}-1)R_B\big(\sqrt{me^{-\gamma}}-1\big)
=R_B\Big(\big(e^{\gamma/2}-e^{-\gamma/2}\big)\sqrt{m}-(e^{\gamma}-1)\Big).
\]
Therefore,
\[
V_{\mathrm{post}}
=R_B\Big(2e^{\gamma/2}\sqrt{m}-e^{\gamma}+1\Big).
\]
The pre-correction AMM value is $V_{\mathrm{pre}}:=PR_A+R_B=R_B(m+1)$, hence the proportional jump return is
\[
J^{-}(m;\gamma)=\frac{V_{\mathrm{post}}}{V_{\mathrm{pre}}}-1
=\frac{2e^{\gamma/2}\sqrt{m}-m-e^{\gamma}}{m+1}
=-\frac{(\sqrt{m}-e^{\gamma/2})^2}{m+1}.
\]

\textbf{Sell-side correction.}
A sell-side correction sets $Q'=Pe^{+\gamma}$ with $K$ fixed, hence
\[
R^{A'}=\sqrt{\frac{K}{Q'}}=R_A\sqrt{\frac{Q_-}{Pe^{\gamma}}}=R_A\frac{e^{-\gamma/2}}{\sqrt{m}},
\qquad
R^{B'}=\sqrt{KQ'}=R_B\sqrt{me^{\gamma}}.
\]
Now $R^{B'}-R_B<0$ and the total $B$ paid out by the pool is $y_{\mathrm{tot}}:=R_B-R^{B'}$. Under the fee rule in log units,
the trader receives $e^{-\gamma}y_{\mathrm{tot}}$ and the paid-out fee to the LP is $y_{\mathrm{fee}}:=(1-e^{-\gamma})y_{\mathrm{tot}}$.
Thus
\[
V_{\mathrm{post}}:=PR^{A'} + R^{B'} + y_{\mathrm{fee}}
=PR^{A'} + R^{B'} + (1-e^{-\gamma})(R_B-R^{B'}).
\]
Using $R_AP=R_B m$ and $R_B-R^{B'}=R_B(1-\sqrt{me^{\gamma}})$ yields
\[
V_{\mathrm{post}}
=R_B\Big(2e^{-\gamma/2}\sqrt{m}+1-e^{-\gamma}\Big),
\]
so with $V_{\mathrm{pre}}=R_B(m+1)$,
\[
J^{+}(m;\gamma)
=\frac{V_{\mathrm{post}}}{V_{\mathrm{pre}}}-1
=\frac{2e^{-\gamma/2}\sqrt{m}-m-e^{-\gamma}}{m+1}
=-\frac{(\sqrt{m}-e^{-\gamma/2})^2}{m+1}.
\]
\end{proof}

\subsubsection{Full Proof of Theorem~\ref{thm:crra_hjb_cfmm}}
\label{app:proof_thm_crra_baseline}

\begin{proof}[Proof of Theorem~\ref{thm:crra_hjb_cfmm}]

Work on a filtered probability space supporting a Brownian motion $B$ and two independent marked Poisson processes
$\{(N_t^{-},M_k^{-})\}$ and $\{(N_t^{+},M_k^{+})\}$ with intensities $\lambda^{-},\lambda^{+}$ and mark laws equal to the
conditional laws of $M^{-},M^{+}$. At a $-$ (resp. $+$) jump, wealth is multiplied by $1+\theta J^{-}(M^{-};\gamma)$
(resp. $1+\theta J^{+}(M^{+};\gamma)$). The admissibility condition in the theorem,
\[
1+\theta J^{-}(M^{-};\gamma)>0\ \text{a.s.},\qquad 1+\theta J^{+}(M^{+};\gamma)>0\ \text{a.s.},
\]
ensures the jump multipliers are strictly positive, so $W_t>0$ a.s. whenever $W_0>0$.

Fix a constant control $\theta\in\Theta_{\mathrm{adm}}$ and let $W$ follow~\ref{eq:wealth_dynamics}. For any $f\in C^2((0,\infty))$
with suitable growth, the infinitesimal generator of $W$ is
\begin{align}
\label{eq:baseline_lagr}
(\mathcal{L}^{\theta} f)(w)
&=
w(\mu-r^{*}\theta) f'(w)
+\frac12 w^2\sigma^2 f''(w)
+\lambda^{-}\Big(\mathbb{E}\big[f(w(1+\theta J^{-}(M^{-};\gamma)))\big]-f(w)\Big)
\nonumber\\
&\quad
+\lambda^{+}\Big(\mathbb{E}\big[f(w(1+\theta J^{+}(M^{+};\gamma)))\big]-f(w)\Big).
\end{align}

The objective is
\[
V(w)=\sup_{\{\theta_t\}}\ \mathbb{E}\Big[\int_0^\infty e^{-\rho t}\frac{W_t^{1-\eta}}{1-\eta}\,dt\ \Big|\ W_0=w\Big].
\]
Standard dynamic programming for discounted infinite-horizon control of jump-diffusions yields the HJB equation:
\begin{equation}
\label{eq:HJB}
0=\sup_{\theta\in\Theta_{\mathrm{adm}}}\left\{
\frac{w^{1-\eta}}{1-\eta}
+(\mathcal{L}^{\theta}V)(w)
-\rho V(w)
\right\}.
\end{equation}

Consider the candidate value function
\[
V(w)=\frac{C}{1-\eta} w^{1-\eta},
\qquad C>0.
\]
Then
\[
V'(w)=Cw^{-\eta},\qquad V''(w)=-\eta C w^{-\eta-1}.
\]
Substituting $V$ into \ref{eq:baseline_lagr} gives
\begin{align*}
(\mathcal{L}^{\theta}V)(w)
&=
w(\mu-r^{*}\theta)\,Cw^{-\eta}
+\frac12 w^2\sigma^2(-\eta C w^{-\eta-1})\\
&\quad
+\lambda^{-}\left(\frac{C}{1-\eta}w^{1-\eta}\mathbb{E}\Big[(1+\theta J^{-}(M^{-};\gamma))^{1-\eta}\Big]-\frac{C}{1-\eta}w^{1-\eta}\right)\\
&\quad
+\lambda^{+}\left(\frac{C}{1-\eta}w^{1-\eta}\mathbb{E}\Big[(1+\theta J^{+}(M^{+};\gamma))^{1-\eta}\Big]-\frac{C}{1-\eta}w^{1-\eta}\right)\\
&=
C w^{1-\eta}\left(\mu-r^{*}\theta-\frac{\eta}{2}\sigma^2\right)
+\frac{C}{1-\eta} w^{1-\eta}\lambda^{-}\Big(\mathbb{E}[(1+\theta J^{-})^{1-\eta}]-1\Big)\\
&\quad
+\frac{C}{1-\eta} w^{1-\eta}\lambda^{+}\Big(\mathbb{E}[(1+\theta J^{+})^{1-\eta}]-1\Big),
\end{align*}
where $J^\pm$ abbreviates $J^\pm(M^\pm;\gamma)$.
Insert this into (\ref{eq:HJB}) and divide by $\frac{C}{1-\eta}w^{1-\eta}$ to obtain
\[
0=\sup_{\theta\in\Theta_{\mathrm{adm}}}\left\{
\frac{1}{C}
-\rho
+(1-\eta)\left(\mu-r^{*}\theta-\frac{\eta}{2}\sigma^2\right)
+\lambda^{-}\Big(\mathbb{E}[(1+\theta J^{-})^{1-\eta}]-1\Big)
+\lambda^{+}\Big(\mathbb{E}[(1+\theta J^{+})^{1-\eta}]-1\Big)
\right\}.
\]
Equivalently,
\[
\frac{1}{C}=\rho-(1-\eta)\sup_{\theta\in\Theta_{\mathrm{adm}}}\Phi(\theta),
\]
where $\Phi(\theta)$ is exactly the function defined in the theorem. This proves the claimed formula for $C$.
The finiteness condition $\rho>(1-\eta)\Phi(\theta^{*})$ implies $C>0$.

We show $\Phi$ is strictly concave on $\Theta_{\mathrm{adm}}$.
Take first and second differentiation of $\Phi$ we get:
\begin{align*}
\Phi'(\theta)
&=
-r^{*}
+\lambda^{-}\mathbb{E}\Big[J^{-}(M^{-};\gamma)(1+\theta J^{-}(M^{-};\gamma))^{-\eta}\Big]
+\lambda^{+}\mathbb{E}\Big[J^{+}(M^{+};\gamma)(1+\theta J^{+}(M^{+};\gamma))^{-\eta}\Big],
\end{align*}

\begin{align*}
\Phi''(\theta)
&=
-\eta\lambda^{-}\mathbb{E}\Big[(J^{-}(M^{-};\gamma))^{2}(1+\theta J^{-}(M^{-};\gamma))^{-\eta-1}\Big]\\
&\quad
-\eta\lambda^{+}\mathbb{E}\Big[(J^{+}(M^{+};\gamma))^{2}(1+\theta J^{+}(M^{+};\gamma))^{-\eta-1}\Big].
\end{align*}
Each expectation is nonnegative and, unless $J^{-}\equiv 0$ a.s.\ and $J^{+}\equiv 0$ a.s., at least one term is strictly
positive; with $\eta>0$ it follows that $\Phi''(\theta)<0$ on $\Theta_{\mathrm{adm}}$, so $\Phi$ is strictly concave.
Because $\Theta_{\mathrm{adm}}$ is compact and $\Phi$ is continuous on it, an optimizer exists; strict concavity yields
uniqueness. If the unique optimizer is interior, it satisfies $\Phi'(\theta^{*})=0$, which is the FOC stated in the theorem.

Then we verify that the candidate value function is in fact optimal, let $\theta^{*}$ denote the unique maximizer of $\Phi$ on $\Theta_{\mathrm{adm}}$ and define
\[
V(w)=\frac{C}{1-\eta}w^{1-\eta}
\quad\text{with}\quad
\frac{1}{C}=\rho-(1-\eta)\Phi(\theta^{*}).
\]
For any admissible control $\{\theta_t\}$, define the process
\[
\mathcal{M}_t
:=
e^{-\rho t}V(W_t)+\int_0^t e^{-\rho s}\frac{W_s^{1-\eta}}{1-\eta}\,ds.
\]
Apply It\^o's formula for jump-diffusions to $e^{-\rho t}V(W_t)$ (using the generator~\ref{eq:baseline_lagr} and the compensator of
the marked Poisson processes). One obtains the decomposition
\[
d\mathcal{M}_t
=
e^{-\rho t}\Big(\frac{W_t^{1-\eta}}{1-\eta}+(\mathcal{L}^{\theta_t}V)(W_t)-\rho V(W_t)\Big)\,dt
+ d(\text{local martingale}).
\]
By the HJB inequality~\ref{eq:HJB}, for every $w$ and admissible $\theta$,
\[
\frac{w^{1-\eta}}{1-\eta}+(\mathcal{L}^{\theta}V)(w)-\rho V(w)\le 0,
\]
hence the drift term of $\mathcal{M}_t$ is nonpositive under any admissible control. Therefore $\mathcal{M}_t$ is a
supermartingale and
\[
V(w)=\mathcal{M}_0\ge \mathbb{E}[\mathcal{M}_t]
=\mathbb{E}\left[\int_0^t e^{-\rho s}\frac{W_s^{1-\eta}}{1-\eta}\,ds+e^{-\rho t}V(W_t)\right].
\]
To pass to $t\to\infty$, we verify transversality. Under any constant control $\theta\in\Theta_{\mathrm{adm}}$, applying
It\^o's formula to $W_t^{1-\eta}$ yields
\[
\frac{d(W_t^{1-\eta})}{W_{t-}^{1-\eta}}
=(1-\eta)\Big(\mu-r^{*}\theta-\frac{\eta}{2}\sigma^2\Big)\,dt
+\big((1+\theta J^{-})^{1-\eta}-1\big)\,dN_t^{-}
+\big((1+\theta J^{+})^{1-\eta}-1\big)\,dN_t^{+}
+(1-\eta)\sigma\,dB_t,
\]
and taking expectations gives
\[
\mathbb{E}[W_t^{1-\eta}]
=w^{1-\eta}\exp\big((1-\eta)\Phi(\theta)\,t\big).
\]
Hence for the candidate $V$,
\[
\mathbb{E}\big[e^{-\rho t}V(W_t)\big]
=\frac{C}{1-\eta}\,w^{1-\eta}\exp\big(((1-\eta)\Phi(\theta)-\rho)t\big)
\to 0
\quad\text{whenever}\quad
\rho>(1-\eta)\Phi(\theta).
\]
In particular, the condition $\rho>(1-\eta)\Phi(\theta^{*})$ guarantees
$\lim_{t\to\infty}\mathbb{E}[e^{-\rho t}V(W_t^{\theta^{*}})]=0$ under $\theta^{*}$.
Therefore letting $t\to\infty$ yields
\[
V(w)\ge \sup_{\{\theta_t\}}\ \mathbb{E}\left[\int_0^\infty e^{-\rho s}\frac{W_s^{1-\eta}}{1-\eta}\,ds\right].
\]

Finally, under the constant feedback control $\theta_t\equiv \theta^{*}$, the HJB is attained with equality, so the drift
of $\mathcal{M}_t$ is identically zero and $\mathcal{M}_t$ is a martingale. Taking expectations and letting $t\to\infty$
using the transversality just shown gives
\[
V(w)=\mathbb{E}\left[\int_0^\infty e^{-\rho s}\frac{W_s^{1-\eta}}{1-\eta}\,ds\right],
\]
so $V$ is the value function and $\theta^{*}$ is optimal. 

If $r^{*} > 0$ and $\theta_{\min}\ge 0$. Then, since Lemma~\ref{lem:Jpm} implies
$J^{-}(M^{-};\gamma)<0$ and $J^{+}(M^{+};\gamma)<0$ almost surely, we have for all $\theta\in\Theta_{\mathrm{adm}}$:
\[
\Phi'(\theta)
=
-r^{*}
+\lambda^{-}\mathbb{E}\!\left[J^{-}(M^{-};\gamma)\,(1+\theta J^{-}(M^{-};\gamma))^{-\eta}\right]
+\lambda^{+}\mathbb{E}\!\left[J^{+}(M^{+};\gamma)\,(1+\theta J^{+}(M^{+};\gamma))^{-\eta}\right]
\;<\; -r^{*} \;<\;0,
\]
hence the unique optimizer is $\theta^{*}=\theta_{\min}$ (in particular, if $\theta_{\min}=0$, then $\theta^{*}=0$).
\end{proof}

\subsubsection{Proof for Lemma~\ref{lem:noise_trader_fee}}
\label{app:proof_for_noise_trader}
\begin{proof}
The unique optimizer for noise trader's optimization problem is given by
\begin{equation}
\label{eq:q_noise_star_sym}
q_t^{N, *}=K_t^\beta \xi_{N, t}
\end{equation}
In particular, conditional on $K_t^{\beta}$, the entry condition is
\begin{equation}
\label{eq:noise_entry}
\frac{1}{2} K_t^\beta \xi_{N, t}^2 \geq g\left(v_t\right) \Longleftrightarrow\left|\xi_{N, t}\right| \geq \sqrt{\frac{2 g\left(v_t\right)}{K_t^\beta}} .
\end{equation}
Hence the realized noise trade is
$$
q_t^N=K_t^\beta \xi_{N, t} \cdot \mathbf{1}\left\{\left|\xi_{N, t}\right| \geq \sqrt{\frac{2 g\left(v_t\right)}{K_t^\beta}}\right\} .
$$
Moreover, conditional on $\mathcal{F}_{t-}$, the buy and sell volumes are symmetric:
$$
\mathbb{E}\left[\left(q_t^N\right)^{+} \mid \mathcal{F}_{t-}\right]=\mathbb{E}\left[\left(q_t^N\right)^{-} \mid \mathcal{F}_{t-}\right]=\frac{1}{2} \mathbb{E}\left[\left|q_t^N\right| \mid \mathcal{F}_{t-}\right] .
$$
Then the instantaneous paid-out noise fee rate (in units of $B$ ):
$$
\dot{F}_t^N=Q_t\left[\left(e^\gamma-1\right)\left(q_t^N\right)^{+}+\left(1-e^{-\gamma}\right)\left(q_t^N\right)^{-}\right] .
$$
Then conditional on $\mathcal{F}_{t-}$, symmetry implies
$$
\mathbb{E}\left[\dot{F}_t^N \mid \mathcal{F}_{t-}\right]=Q_t \cdot \frac{\left(e^\gamma-e^{-\gamma}\right)}{2} \cdot \mathbb{E}\left[\left|q_t^N\right| \mid \mathcal{F}_{t-}\right] .
$$
Now compute $\mathbb{E}\left|q_t^N\right|$ under the threshold rule. Since $q_t^N=K_t^\beta \xi 1\left\{|\xi| \geq a_t\right\}$ with $a_t:=\sqrt{\frac{2 g\left(v_t\right)}{K_t^\beta}},$ and $\xi \sim \mathcal{N}\left(0, \sigma_N^2\right)$, we have $\mathbb{E}\left[\left|q_t^N\right| \mid \mathcal{F}_{t-}\right]=K_t^\beta \mathbb{E}\left[|\xi| \mathbf{1}\left\{|\xi| \geq a_t\right\}\right] .$
A standard truncated-normal identity gives $\mathbb{E}\left[|\xi| \mathbf{1}\left\{|\xi| \geq a_t\right\}\right]=2 \sigma_N \varphi\left(\frac{a_t}{\sigma_N}\right),
$
where $\varphi$ is the standard normal pdf. Therefore
$$
\mathbb{E}\left[\dot{F}_t^N \mid \mathcal{F}_{t-}\right]=Q_t \cdot\left(e^\gamma-e^{-\gamma}\right) \cdot K_t^\beta \sigma_N \varphi\left(\frac{1}{\sigma_N} \sqrt{\frac{2 g\left(v_t\right)}{K_t^\beta}}\right) .
$$
Moreover, for any fixed $K>0$, if $g'(v)>0$ then
\begin{equation}
\label{eq:noise_fee_v_derivative}
\frac{\partial}{\partial v}\,
\mathbb E[\dot F_t^{N}\mid K_t=K, Q_t]
=
-\,Q_t\cdot (e^\gamma-e^{-\gamma})\cdot \frac{g'(v)}{\sigma_N}\,\varphi\!\big(z(v,K)\big)
\;<\;0.
\end{equation}
\end{proof}

\subsubsection{Proof for Proposition~\ref{prop:race_overrun}}
\label{app:proof_for_race_overrun}
\begin{proof}

Fix $\Delta_t>0$ and pool depth $K_t^{\beta}$. For a given belief $\hat N_i$, substitute $\pi_i\approx 1/\hat N_i$ into the entrant's
objective:
\[
\max_{q_i\ge 0}\ \frac{1}{\hat N_i}q_i\Delta_t-\frac{q_i^2}{2K_t^{\beta}}-g_t.
\]
This is a concave quadratic in $q_i$. The first-order condition yields
\[
\frac{1}{\hat N_i}\Delta_t-\frac{q_i}{K_t^{\beta}}=0
\quad\Rightarrow\quad
q_i^*=K_t^{\beta}\frac{\Delta_t}{\hat N_i},
\]

The corresponding maximal expected net payoff equals
\[
\mathbb E[\Pi_i(q_i^*)\mid \Delta_t,\hat N_i]=\frac{K_t^{\beta}\,\Delta_t^2}{2\hat N_i^2}-g_t,
\]
so agent $i$ enters iff
\begin{equation}
\label{eq:entry_condition}
|\Delta_t|\ \ge\ \sqrt{\frac{2g_t}{K_t^{\beta}}}\ \hat N_i.
\end{equation}

Given a realized set of entrants and winner $w$, the total overrun quantity is the sum of overrun arbitrageurs' submitted quantities:
\[
L_A=\sum_{i\neq w} q_i^*
=K_t^{\beta}\Delta_t \sum_{i\neq w}\frac{1}{\hat N_i},
\]
In the unbiased benchmark $\hat N_i=N$, this becomes
$L_A\approx K_t^{\beta}\Delta_t(1-1/N)$, implying the displayed conditional expectation.

Finally, since $\Delta_t\mid v_t\sim\mathcal N(0,v_t)$, we have
$\mathbb E[|\Delta_t|\mid v_t]=\sqrt{\frac{2}{\pi}}\sqrt{v_t}$.
Averaging $K_t^{\beta}(1-1/N)\Delta_t$ over $|\Delta_t|$ and over the induced entry/crowding produces scaling factor summarized by $\kappa_{\mathrm{bel}}\in[\frac12,1]$ and larger under
overconfidence due to convexity of $\frac{1}{x}$, for simplicity we will take $\kappa_{bel} = 1$, and in the future work we can explore  $\kappa_{bel}$ that depends on exogenous variables like $v_t$.
\end{proof}

\subsubsection{Proof for Lemma~\ref{lem:Jext}}
\label{app:proof_for_Jext}
\begin{proof}
The derivation is the same up to the winner of the race correcting the price, write $U^- = u$, after the overrun flow, the new reserves are 
\[
\begin{aligned}
& R_A^{+}=(1-u) R_A^b \\
& R_B^{+}=R_B^b /\left(1- u\right)
\end{aligned}
\]
The inventory value becomes 
$$
V_{+}^{\mathrm{inv}}=P R_A^{+}+R_B^{+}=P(1-u) R_A^b+\frac{R_B^b}{1-u} .
$$
Now express $P R_A^b$ and $R_B^b$ in a common unit. From boundary condition $R_B^b=Q^b R_A^b=P e^{-\gamma} R_A^b$, we have
$$
P R_A^b=e^\gamma R_B^b .
$$
So
$$
V_{+}^{\mathrm{inv}}=(1-u) e^\gamma R_B^b+\frac{R_B^b}{1-u}=R_B^b\left((1-u) e^\gamma+\frac{1}{1-u}\right) .
$$
Use $R_B^b=R_B \sqrt{m} e^{-\gamma / 2}$ :
$$
V_{+}^{\mathrm{inv}}=R_B \sqrt{m} e^{-\gamma / 2}\left((1-u) e^\gamma+\frac{1}{1-u}\right)=R_B \sqrt{m}\left((1-u) e^{\gamma / 2}+\frac{e^{-\gamma / 2}}{1-u}\right) .
$$
Adding the fee term $\left(e^\gamma-1\right)\left(\frac{R_B \sqrt{m} e^{-\gamma / 2}}{1-u}-R_B\right)$, we have
\[
V_{+}^{\mathrm{tot}}=R_B\left(e^{\gamma / 2} \sqrt{m}\left((1-u)+\frac{1}{1-u}\right)+1-e^\gamma\right)
\]
and 
\[
J_{\mathrm{ext}}^{-}(m, u ; \gamma)=\frac{V_+^{\mathrm{tot}}}{V^-} = \frac{e^{\gamma / 2} \sqrt{m}\left((1-u)+\frac{1}{1-u}\right)-m-e^\gamma}{m+1} .
\]
Since $1-u + \frac{1}{1-u} = \frac{u^2-2u+2}{1-u} = 1+\frac{u^2}{1-u}$, we have 
\[
J_{\mathrm{ext}}^{-}(m, u ; \gamma) = \frac{e^{-\gamma / 2} \sqrt{m} \cdot \frac{u^2}{1-u}-\left(e^{-\gamma / 2}-\sqrt{m}\right)^2}{m+1}
\]
Differentiating in $u$ yields
\[
\partial_u J_{\mathrm{ext}}^{-}(m,u;\gamma)
=
\frac{e^{\gamma/2}\sqrt{m}}{m+1}\cdot\frac{u(2-u)}{(1-u)^2}>0,
\qquad
\partial_{uu}J_{\mathrm{ext}}^{-}(m,u;\gamma)
=
\frac{2e^{\gamma/2}\sqrt{m}}{m+1}\cdot\frac{1}{(1-u)^3}>0,
\]
\end{proof}

\subsubsection{Proof for Theorem~\ref{thm:crra_extended_endogK}}
\label{app:full_thm_and_proof_extended}

\begin{proposition}[Existence, uniqueness, and interiority of the reduced-form optimum]
\label{prop:existence_interior_clean}
Fix time $t$ and condition on $\mathcal F_{t-}$. Let $\theta\in[\theta_{\min},\theta_{\max}]$ and consider the reduced objective
\begin{align}
\label{eq:Phi_clean}
\Phi_t(\theta)
&=
\mu-r^*\theta-\frac{\eta}{2}v_t
+\theta\,\psi_t^N(\theta)
+\sum_{\pm}\frac{\lambda_t^\pm(\theta,v_t)}{1-\eta}\,
\mathbb E\!\left[(1+\theta J_t^\pm(\theta))^{1-\eta}-1\ \big|\ \mathcal F_{t-}\right],
\end{align}
where $J_t^\pm(\theta)=J_{\mathrm{ext}}^\pm(m_t,U_t^\pm(\theta);\gamma)$ and $\psi_t^N(\theta)$ is the proportional noise-fee yield.

Define the admissible set
\begin{equation}
\Theta_{\mathrm{adm}}
=
\Big\{\theta\in[\theta_{\min},\theta_{\max}]:
1+\theta J_t^-(\theta)>0\ \text{a.s.},\;
1+\theta J_t^+(\theta)>0\ \text{a.s.}
\Big\}.
\end{equation}

Assume:
\begin{enumerate}
\item[\textbf{(A1)}] (\emph{Compact admissible set}) $\Theta_{\mathrm{adm}}$ is nonempty and compact.
\item[\textbf{(A2)}] (\emph{Regularity}) $\psi_t^N(\theta)$, $J_t^\pm(\theta)$, and $\lambda_t^\pm(\theta,v_t)$ are twice continuously
differentiable on $\Theta_{\mathrm{adm}}$, and differentiation may be passed under the expectation.
\item[\textbf{(A3)}] (\emph{Local concavity}) $\Phi_t$ is strictly concave on $\Theta_{\mathrm{adm}}$.
\end{enumerate}

Then:
\begin{enumerate}
\item[\textbf{(i)}] (\emph{Existence and uniqueness})
There exists a unique maximizer
$\theta_t^*\in\arg\max_{\theta\in\Theta_{\mathrm{adm}}}\Phi_t(\theta)$.

\item[\textbf{(ii)}] (\emph{First-order condition})
$\Phi_t$ is differentiable on $\mathrm{int}(\Theta_{\mathrm{adm}})$, and any interior maximizer satisfies
$\Phi_t'(\theta_t^*)=0$, where
\begin{align}
\label{eq:Phi_prime_clean}
\Phi_t'(\theta)
&=
-r^*
+\psi_t^N(\theta)+\theta\,\psi_{\theta,t}^N(\theta)
+\sum_{\pm}\Bigg[
\lambda_t^\pm(\theta,v_t)\,H_t^\pm(\theta)
+\frac{\lambda_{\theta,t}^\pm(\theta,v_t)}{1-\eta}\,G_t^\pm(\theta)
\Bigg],
\end{align}
with
\begin{align*}
G_t^\pm(\theta)
&:=\mathbb E\!\left[(1+\theta J_t^\pm(\theta))^{1-\eta}-1\ \big|\ \mathcal F_{t-}\right],\\
H_t^\pm(\theta)
&:=\mathbb E\!\left[
\big(J_t^\pm(\theta)+\theta J_{\theta,t}^\pm(\theta)\big)
(1+\theta J_t^\pm(\theta))^{-\eta}
\ \big|\ \mathcal F_{t-}\right].
\end{align*}

\item[\textbf{(iii)}] (\emph{Interiority test})
The maximizer is interior if and only if the one-sided derivatives satisfy
\[
\Phi_t'(\theta_{\min})>0,
\qquad
\Phi_t'(\theta_{\max})<0.
\]
\end{enumerate}
\end{proposition}

\begin{proof}
\textbf{(i) Existence and uniqueness.}
Under \textbf{(A2)}, $\Phi_t$ is continuous on $\Theta_{\mathrm{adm}}$.
Under \textbf{(A1)}, the domain is compact, hence $\Phi_t$ attains a maximizer.
Under \textbf{(A3)}, strict concavity implies the maximizer is unique.

\medskip
\textbf{(ii) First-order condition.}
By \textbf{(A2)}, $\Phi_t$ is differentiable on $\mathrm{int}(\Theta_{\mathrm{adm}})$.
Differentiating \eqref{eq:Phi_clean} yields \eqref{eq:Phi_prime_clean}.
Any interior maximizer $\theta_t^*$ must satisfy $\Phi_t'(\theta_t^*)=0$.

\medskip
\textbf{(iii) Interiority.}
Under strict concavity, $\Phi_t'$ is strictly decreasing.
Therefore the maximizer lies in the interior if and only if the marginal net benefit is positive at
$\theta_{\min}$ and negative at $\theta_{\max}$, which gives the stated condition.

\medskip
\textbf{(iv) Structure of derivatives.}
The derivative \eqref{eq:Phi_prime_clean} decomposes naturally into:
(i) the financing cost $r^*$,
(ii) the marginal noise-fee benefit $\psi_t^N+\theta\psi_{\theta,t}^N$,
(iii) the jump-severity channel $\lambda^\pm H_t^\pm$,
and (iv) the frequency channel $\lambda_{\theta,t}^\pm G_t^\pm/(1-\eta)$.

When overrun is small, $J_t^\pm(\theta)=O(U_t^\pm(\theta)^2)$ and $J_{\theta,t}^\pm(\theta)=O(U_t^\pm(\theta))$,
so $H_t^\pm(\theta)$ is dominated by the curvature term.
This ensures local concavity provided that the sensitivity of $\lambda_t^\pm(\theta,v)$ and $\psi_t^N(\theta)$
to pool scaling is sufficiently weak, as formalized in \textbf{(A3)}.
\end{proof}

\begin{proposition}[Sufficient conditions for strict concavity of $\Phi_t$ under endogenous $\lambda^\pm(\theta,v)$]
\label{prop:concavity_clean}
Fix time $t$ and condition on $\mathcal F_{t-}$. Let $\Theta_{\mathrm{adm}}$ be defined as in
\eqref{eq:Theta_adm_oneprop}, and consider the reduced objective
\begin{align}
\label{eq:Phi_concave}
\Phi_t(\theta)
&=
\mu-r^*\theta-\frac{\eta}{2}v_t
+\theta\,\psi_t^N(\theta)
+\sum_{\pm}\frac{\lambda_t^\pm(\theta,v_t)}{1-\eta}\,
\mathbb E\!\left[(1+\theta J_t^\pm(\theta))^{1-\eta}-1\ \big|\ \mathcal F_{t-}\right],
\end{align}
where $J_t^\pm(\theta)=J_{\mathrm{ext}}^\pm(m_t,U_t^\pm(\theta);\gamma)$ and $\psi_t^N(\theta)$ is the proportional noise-fee yield.

Assume the regularity in \textbf{(A2)} of Proposition~\ref{prop:existence_interior_clean} and that $\eta>1$.
Define, for each sign $\pm$, the expectation functionals
\begin{align}
\label{eq:G_H_def}
G_t^\pm(\theta)
&:=\mathbb E\!\left[(1+\theta J_t^\pm(\theta))^{1-\eta}-1\ \big|\ \mathcal F_{t-}\right],\\
H_t^\pm(\theta)
&:=\mathbb E\!\left[
\big(J_t^\pm(\theta)+\theta J_{\theta,t}^\pm(\theta)\big)
(1+\theta J_t^\pm(\theta))^{-\eta}
\ \big|\ \mathcal F_{t-}\right],\nonumber\\
K_t^\pm(\theta)
&:=\mathbb E\!\left[
\big(2J_{\theta,t}^\pm(\theta)+\theta J_{\theta\theta,t}^\pm(\theta)\big)
(1+\theta J_t^\pm(\theta))^{-\eta}
-\eta\big(J_t^\pm(\theta)+\theta J_{\theta,t}^\pm(\theta)\big)^2
(1+\theta J_t^\pm(\theta))^{-\eta-1}
\ \big|\ \mathcal F_{t-}\right].\nonumber
\end{align}

Then $\Phi_t$ is twice differentiable on $\mathrm{int}(\Theta_{\mathrm{adm}})$ and satisfies
\begin{align}
\label{eq:Phi_second_full}
\Phi_t''(\theta)
&=
2\psi_{\theta,t}^N(\theta)+\theta\psi_{\theta\theta,t}^N(\theta)
+\sum_{\pm}\Bigg[
\lambda_t^\pm(\theta,v_t)\,K_t^\pm(\theta)
+2\lambda_{\theta,t}^\pm(\theta,v_t)\,H_t^\pm(\theta)
+\frac{\lambda_{\theta\theta,t}^\pm(\theta,v_t)}{1-\eta}\,G_t^\pm(\theta)
\Bigg].
\end{align}

Moreover, a sufficient condition for strict concavity on $\Theta_{\mathrm{adm}}$ is that, for all $\theta\in\Theta_{\mathrm{adm}}$,
\begin{equation}
\label{eq:concavity_suff}
\eta\sum_{\pm}\lambda_t^\pm(\theta,v_t)\,
\underbrace{\mathbb E\!\left[
\big(J_t^\pm(\theta)+\theta J_{\theta,t}^\pm(\theta)\big)^2
(1+\theta J_t^\pm(\theta))^{-\eta-1}
\ \big|\ \mathcal F_{t-}\right]}_{=:A_t^\pm(\theta)}
\;>\;
B_t(\theta)+N_t(\theta),
\end{equation}
where
\begin{align}
\label{eq:concavity_BN}
B_t(\theta)
&:=
\sum_{\pm}\lambda_t^\pm(\theta,v_t)\,
\underbrace{\mathbb E\!\left[
\big|2J_{\theta,t}^\pm(\theta)+\theta J_{\theta\theta,t}^\pm(\theta)\big|
(1+\theta J_t^\pm(\theta))^{-\eta}
\ \big|\ \mathcal F_{t-}\right]}_{=:B_t^\pm(\theta)},\\
N_t(\theta)
&:=
\Big|2\psi_{\theta,t}^N(\theta)+\theta\psi_{\theta\theta,t}^N(\theta)\Big|
+\sum_{\pm}2\big|\lambda_{\theta,t}^\pm(\theta,v_t)\big|\cdot
\underbrace{\mathbb E\!\left[
\big|J_t^\pm(\theta)+\theta J_{\theta,t}^\pm(\theta)\big|
(1+\theta J_t^\pm(\theta))^{-\eta}
\ \big|\ \mathcal F_{t-}\right]}_{=:C_t^\pm(\theta)} \nonumber\\
&\quad
+\sum_{\pm}\Big|\frac{\lambda_{\theta\theta,t}^\pm(\theta,v_t)}{1-\eta}\Big|\cdot
\big|G_t^\pm(\theta)\big|.\nonumber
\end{align}

In particular, under $\eta>1$ we have $G_t^\pm(\theta)\le 0$ for all admissible $\theta$, so the last term in $N_t(\theta)$ is controlled
by the magnitude of $\lambda_{\theta\theta,t}^\pm(\theta,v_t)$.

Finally, suppose $K(\theta)=\bar K\theta^2$ and $\lambda_t^\pm(\theta,v)$ is given by \eqref{eq:lambda_gas_monotone}, i.e.
\[
\lambda_t^\pm(\theta,v)=2\bar\lambda^\pm\Big(1-\Phi(z(\theta,v))\Big),
\qquad z(\theta,v):=\sqrt{\frac{2g(v)}{v\,(\bar K\theta^2)^\beta}}.
\]
Then
\begin{align}
\label{eq:lambda_derivs_closed}
\lambda_{\theta,t}^\pm(\theta,v)
&=
2\beta\,\bar\lambda^\pm\;\phi(z)\;\frac{z}{\theta}>0,\\
\lambda_{\theta\theta,t}^\pm(\theta,v)
&=
2\beta\,\bar\lambda^\pm\;\phi(z)\;\frac{z}{\theta^2}\Big(\beta z^2-(\beta+1)\Big),
\end{align}
so a convenient sufficient condition to keep the $\lambda_{\theta\theta}$-term from harming concavity is the ``active-region'' bound
\begin{equation}
\label{eq:active_region}
z(\theta,v)^2 \le 1+\frac{1}{\beta}
\quad\Longrightarrow\quad
\lambda_{\theta\theta,t}^\pm(\theta,v)\le 0.
\end{equation}
\end{proposition}

\begin{proof}
Differentiate $\Phi_t(\theta)$ in \eqref{eq:Phi_concave}.
The fee component gives
$\frac{d}{d\theta}(\theta\psi_t^N)=\psi_t^N+\theta\psi_{\theta,t}^N$ and
$\frac{d^2}{d\theta^2}(\theta\psi_t^N)=2\psi_{\theta,t}^N+\theta\psi_{\theta\theta,t}^N$.

For each sign $\pm$, write the jump block as
\[
\Phi_t^{J,\pm}(\theta)=\frac{\lambda_t^\pm(\theta,v_t)}{1-\eta}G_t^\pm(\theta).
\]
By the chain rule and the admissibility $1+\theta J_t^\pm(\theta)>0$ a.s.,
\[
\frac{d}{d\theta}G_t^\pm(\theta)
=(1-\eta)H_t^\pm(\theta),
\qquad
\frac{d}{d\theta}H_t^\pm(\theta)=K_t^\pm(\theta).
\]
Therefore
\[
\frac{d}{d\theta}\Phi_t^{J,\pm}(\theta)
=
\lambda_t^\pm(\theta,v_t)H_t^\pm(\theta)
+
\frac{\lambda_{\theta,t}^\pm(\theta,v_t)}{1-\eta}G_t^\pm(\theta),
\]
and differentiating again yields the decomposition in \eqref{eq:Phi_second_full}:
\[
\frac{d^2}{d\theta^2}\Phi_t^{J,\pm}(\theta)
=
\lambda_t^\pm K_t^\pm
+2\lambda_{\theta,t}^\pm H_t^\pm
+\frac{\lambda_{\theta\theta,t}^\pm}{1-\eta}G_t^\pm.
\]

To obtain the sufficient concavity condition, decompose $K_t^\pm(\theta)$ in \eqref{eq:G_H_def} into its negative CRRA-curvature term
and its sensitivity term:
\[
K_t^\pm(\theta)
=
-\eta A_t^\pm(\theta)
+\widetilde B_t^\pm(\theta),
\qquad
|\widetilde B_t^\pm(\theta)|\le B_t^\pm(\theta),
\]
where $A_t^\pm,B_t^\pm$ are defined in \eqref{eq:concavity_suff}--\eqref{eq:concavity_BN}.
Substituting this bound into \eqref{eq:Phi_second_full} and applying triangle inequalities to the remaining terms
yields the sufficient condition \eqref{eq:concavity_suff}.

Finally, under $K(\theta)=\bar K\theta^2$ and \eqref{eq:lambda_gas_monotone}, the closed forms in
\eqref{eq:lambda_derivs_closed} follow by differentiating the tail probability
$\mathbb P(|\Delta|\ge \sqrt{2g(v)/K(\theta)^\beta}\mid v)=2(1-\Phi(z(\theta,v)))$ and using $\Phi'(z)=\phi(z)$.
The active-region bound \eqref{eq:active_region} is the condition under which $\beta z^2-(\beta+1)\le 0$, implying
$\lambda_{\theta\theta,t}^\pm(\theta,v)\le 0$.
\end{proof}

\subsubsection{Proof for Theorem~\ref{thm:hump_theta_v}}
\label{app:proof_for_hump}
\begin{proof}
By the implicit function theorem, whenever $F_\theta\neq 0$,
\[
\frac{d\theta^*}{dv} = -\frac{F_v}{F_\theta}.
\]
Under Assumption 1, $F_\theta(\theta^*(v),v)<0$, so the sign of $d\theta^*/dv$ equals the sign of $F_v$.

\emph{(Rising region).} At $v_0$, Assumption 4 gives $F_v(\theta^*(v_0),v_0)>0$, hence $d\theta^*/dv(v_0)>0$.

\emph{(Falling tail).} By Lemma~\ref{lem:noise_trader_fee} and $g(v)\to\infty$, we have
$\mathbb E[\dot F^N\mid \mathcal F_{t-}]\to 0$ and hence $\psi^N(\theta,v)\to 0$ for any fixed $\theta$.
Next, using \eqref{eq:lambda_gas_monotone} with $\Delta\mid v\sim\mathcal N(0,v)$,
\[
\mathbb P\!\left(|\Delta|\ge \sqrt{\frac{2g(v)}{K(\theta)^\beta}}\ \Big|\ v\right)
=
2\Big(1-\Phi\Big(\sqrt{\frac{2g(v)}{K(\theta)^\beta v}}\Big)\Big).
\]
Since $g(v)/v$ is increasing and diverges to $\infty$, the argument of $\Phi(\cdot)$ diverges, so the tail probability and hence
$\lambda^\pm(\theta,v)$ converge to $0$ for any fixed $\theta$.
Therefore the entire jump block in \eqref{eq:FOC_reduced_cs} vanishes as $v\to\infty$ (even though $u(\theta,v)$ may grow, it is multiplied
by $\lambda^\pm(\theta,v)\to 0$).
In the limit, $F(\theta,v)\to -r^*<0$ uniformly over interior $\theta$, so the interior FOC cannot be satisfied at large $v$;
the optimizer is pushed to the lower boundary $\theta_{\min}$.
Hence $\theta^*(v)\to\theta_{\min}$ as $v\to\infty$, implying $\theta^*(v)$ is eventually decreasing.

\emph{(Hump).} Since $\theta^*$ increases near $v_0$ but converges back toward $\theta_{\min}$ as $v\to\infty$, it must attain a maximum at some
finite $v^*$.
\end{proof}

\end{document}